\documentclass[12pt, a4paper, reqno]{amsart}

\usepackage{lipsum}
\usepackage{amsmath}
\usepackage{amssymb}
\usepackage{amsthm}
\usepackage{mathtools}
\usepackage{fancyhdr}
\usepackage{graphicx}
\usepackage{verbatim}
\usepackage{latexsym}
\usepackage{mathrsfs}
\usepackage{mathdots}

\usepackage{tikz}
\usepackage{tikz-cd}

\usepackage{luatex85} % https://tex.stackexchange.com/questions/501699/lualatex-problem-with-xy-pic
\usepackage[all]{xy}

\usepackage{float}
\usepackage{authblk}
\usepackage{eqparbox}

\usepackage{array}   
\usepackage{mathtools}
\usepackage{cellspace}
\usepackage{colortbl}
\usepackage{multicol}

\usepackage[foot]{amsaddr}

\usepackage{todonotes}
\setuptodonotes{inline}

\usepackage[shortlabels]{enumitem}
\usepackage{pdfpages}
\usepackage{xcolor}
\usepackage{stmaryrd}
\usepackage{wallpaper}
\usepackage{setspace}
\usepackage[top = 2 cm, bottom = 3 cm, left = 2 cm, right = 2 cm]{geometry}
\usepackage{bbm}

\usepackage[UKenglish]{isodate}% http://ctan.org/pkg/isodate
\cleanlookdateon% Remove ordinal day reference

\usepackage{anyfontsize}
\usepackage{pgf,pgfplots}
\usepackage{textcomp}
\usepackage[margin = 1cm]{caption}
\usepackage{subcaption}
\usepackage{stackengine}

\usepackage{wasysym}
\usepackage{esint} % makes int ... int slimmer

\usepackage[style=alphabetic,
  sorting = nyt,
  sortcites=true,
  giveninits=true,
  date=year,
  isbn=false,
  maxbibnames=99,
  maxalphanames=6,
  backend=biber]{biblatex}
\DeclareFieldFormat*{titlecase}{\MakeSentenceCase{#1}}

\DeclareSourcemap{\maps[datatype=bibtex]{\map{\step[fieldset=shorthand, null]}}}
\DeclareSourcemap{
  \maps[datatype=bibtex]{
    \map[overwrite]{
      \step[fieldsource=doi, final]
      \step[fieldset=url, null]
      \step[fieldset=eprint, null]
    }  
  }
}

\DeclareSourcemap{
  \maps[datatype=bibtex, overwrite]{
    \map{
      \step[fieldset=language, null]
      \step[fieldset=month, null]
      \step[fieldset=urldate, null]
    }
  }
}

\AtEveryBibitem{%
  \clearlist{language}%
  \clearlist{month}%
  \clearlist{urldate}%
  \clearfield{urldate}%
}

\renewbibmacro*{in:}{%
  \setunit{\addcomma\space}%
  \ifentrytype{article}
    {}
    {\printtext{%
       \bibstring{in}\intitlepunct}}}

\renewcommand*{\intitlepunct}{\addspace}

\DeclareFieldFormat[misc]{title}{\mkbibquote{#1}}
\DeclareFieldFormat[article,inproceedings]{volume}{\mkbibbold{#1}}
\DeclareFieldFormat[inproceedings]{series}{\mkbibitalic{#1}}

\usepackage[hidelinks,unicode,breaklinks=true]{hyperref}
\usepackage[capitalise,noabbrev,sort]{cleveref}

\usepackage{bookmark}

\usepackage{seqsplit}
\usepackage[notcite,notref,final]{showkeys}

\pgfplotsset{compat=1.16}
\usetikzlibrary{arrows}
\usetikzlibrary{arrows.meta}
\usetikzlibrary{patterns}
\usetikzlibrary{calc}
\usetikzlibrary{math} 

\usetikzlibrary{positioning,decorations.pathreplacing} % https://tex.stackexchange.com/questions/82700/set-braces-in-xypic-xymatrix

\usepackage{chngcntr}
\counterwithin{figure}{section}
\counterwithin{equation}{section}
\counterwithin{table}{section}

\let\C\relax

\newcommand{\R}{\mathbb{R}}
\newcommand{\C}{\mathbb{C}}
\newcommand{\Q}{\mathbb{Q}}
\newcommand{\N}{\mathbb{N}}
\newcommand{\Z}{\mathbb{Z}}

\newcommand{\X}{\mathbb{X}}

\newcommand{\T}{\mathbb{T}}
\newcommand{\A}{\mathbb{A}}

\renewcommand{\H}{\mathbb{H}}

\newcommand{\mcN}{\mathcal{N}}

\newcommand{\mcH}{\mathcal{H}}
\newcommand{\mcF}{\mathcal{F}}

\newcommand{\mcE}{\mathcal{E}}

\newtheoremstyle{theorems}% name
  {3pt}%      Space above
  {3pt}%      Space below
  {\itshape}%         Body font
  {}%         Indent amount (empty = no indent, \parindent = para indent)
  {\bfseries}% Thm head font
  {.}%        Punctuation after thm head
  { }%     Space after thm head: " " = normal interword space;
  {}%         Thm head spec (can be left empty, meaning `normal')

\newtheoremstyle{proofparts}% name
  {3pt}%      Space above
  {0pt}%      Space below
  {}%         Body font
  {\parindent}%         Indent amount (empty = no indent, \parindent = para indent)
  {\scshape}% Thm head font
  {:}%        Punctuation after thm head
  {\newline}%     Space after thm head: " " = normal interword space;
  {}%         Thm head spec (can be left empty, meaning `normal')

\newtheoremstyle{claims}% name
  {2pt}%      Space above
  {2pt}%      Space below
  {}%         Body font
  {\parindent}%         Indent amount (empty = no indent, \parindent = para indent)
  {\bfseries}% Thm head font
  {.}%        Punctuation after thm head
  { }%     Space after thm head: " " = normal interword space;
  {}%         Thm head spec (can be left empty, meaning `normal')

\theoremstyle{theorems}
\newtheorem{thm}{Theorem}[section]

\newtheorem{lemma}[thm]{Lemma}
\newtheorem*{lemma*}{Lemma}

\newtheorem{prop}[thm]{Proposition}

\newtheorem*{conj*}{Conjecture}

\theoremstyle{definition}
\newtheorem{defn}[thm]{Definition}

\newtheorem{remark}[thm]{Remark}
\newtheorem{notation}[thm]{Notation}

\theoremstyle{proofparts}

\makeatletter
\@addtoreset{proofpart}{thm}
\makeatother

\theoremstyle{claims}

\newtheorem*{claim*}{Claim}

\crefname{thm}{theorem}{theorems}
\crefname{problem}{problem}{problems}
\crefname{lemma}{lemma}{lemmas}
\crefname{cor}{corollary}{corollaries}
\crefname{prop}{proposition}{propositions}
\crefname{conj}{conjecture}{conjectures}
\crefname{defn}{definition}{definitions}
\crefname{note}{note}{notes}
\crefname{ex}{example}{examples}
\crefname{remark}{remark}{remarks}
\crefname{notation}{notation}{notations}
\crefname{assumption}{assumption}{assumptions}
\crefname{claim}{claim}{claims}
\crefname{claim*}{claim}{claims}

\makeatletter
\newcommand{\Biggg}{\bBigg@{3}}
\newcommand{\vast}{\bBigg@{4}}
\newcommand{\Vast}{\bBigg@{5}}
\makeatother

\newcommand{\norm}[1]{\left\Vert #1 \right\Vert}
\newcommand{\abs}[1]{\left\vert #1 \right\vert}

\DeclareMathOperator{\supp}{supp}

\DeclareMathOperator{\spec}{spec}

\DeclareMathOperator{\tspan}{span}

\let\Re\relax\DeclareMathOperator{\Re}{Re}

\definecolor{emphcolor}{rgb}{0,0,1}           % EMPHASIS COLOUR

\newcommand{\ip}[2]{\left\langle #1 \middle\vert #2 \right\rangle}
\newcommand{\longip}[3]{\left\langle #1 \middle\vert #2 \middle\vert #3 \right\rangle}

\newcommand{\expect}[1]{\left\langle #1 \right\rangle}

\newcommand{\ud}{\,\textnormal{d}}
\newcommand{\dd}[1]{\frac{\textnormal{d}}{\textnormal{d} #1}}

\newcommand{\ket}[1]{\left\vert #1 \right\rangle}

\newcommand{\hc}{\textnormal{h.c.}}

\let\oldepsilon\epsilon
\let\epsilon\varepsilon
\let\varepsilon\oldepsilon

\let\eps\epsilon

\title{Ground state energy of the dilute spin-polarized Fermi gas: Lower bound}
\author{Asbjørn Bækgaard Lauritsen}
\email{alaurits@ist.ac.at}
\author{Robert Seiringer}
\email{rseiring@ist.ac.at}
\address{Institute of Science and Technology Austria, Am Campus 1, 3400 Klosterneuburg, Austria}
\date{\today}

\allowdisplaybreaks

\begin{document}

\begin{abstract}
We prove a lower bound on the ground state energy of the dilute spin-polarized Fermi gas
capturing the leading correction to the kinetic energy resulting from repulsive interactions. 
This correction depends on the $p$-{wave scattering length} of the interaction 
and matches the corresponding upper bound in [J. Funct. Anal. 286.7 (2024), p. 110320].
% [arXiv:2301.04894].
\end{abstract}

\maketitle

\section{Introduction and Main Results}

In recent years much effort in mathematical physics has been devoted towards establishing the validity of asymptotic formulas for the ground state energy (and the free energy at non-zero temperature) of dilute quantum gases, motivated in part by the advances in the physics of cold atoms. The validity of leading order terms was proved for bosons at zero \cite{Dyson.1957,Lieb.Yngvason.1998,Lieb.Yngvason.2001}  and positive  \cite{Seiringer.2008,Yin.2010,Deuchert.Mayer.ea.2020,Mayer.Seiringer.2020} temperature, as well as for fermions with $q\geq 2$ spin components \cite{Lieb.Seiringer.ea.2005,Seiringer.2006,Falconi.Giacomelli.ea.2021}, in dimensions $d\geq 2$. For bosons, even the next-order term (the Lee--Huang--Yang correction) has been established \cite{Yau.Yin.2009,Fournais.Solovej.2020,Basti.Cenatiempo.ea.2021,Fournais.Solovej.2022,Fournais.Girardot.ea.2022,Haberberger.Hainzl.ea.2023}. In all these cases, the strength of the interparticle interaction is quantified by the $s$-wave scattering length. In the special case of one dimension, both bosonic and fermionic gases were recently studied in \cite{Agerskov.Reuvers.ea.2022}. 

Notably absent from this list is the spinless Fermi gas in dimensions $d\geq 2$. At low density, the effect of the interactions (quantified by the $p$-wave scattering length in this case) is significantly smaller than for bosons or for fermions with spin, due to the vanishing of the wave functions at spatial coincidences of the particles as a consequence of the Pauli principle. This makes the mathematical analysis much more subtle. 
In \cite{Lauritsen.Seiringer.2024}, an asymptotically correct upper bound on the ground state energy of a dilute spinless Fermi gas was obtained, by developing a cluster expansion technique (see also \cite{Lauritsen.Seiringer.2023a} for a version of this technique applicable at non-zero temperature). 
In this paper, we prove a corresponding lower bound. 
Our method is inspired by \cite{Falconi.Giacomelli.ea.2021}, and utilizes a suitable unitary transformation implementing the relevant correlations when two particles are at close distances. 

To formulate our result more precisely, consider a gas of $N$ indistinguishable fermions in a $d$-dimensional box $\Lambda = [-L/2,L/2]^d$ of side length $L>0$ 
interacting through a repulsive pair potential $V$, meaning that $V \geq 0$.
Our main focus will be on the physically most relevant case $d=3$.
The Hamiltonian of such a system is given by 
(in units where $\hbar = 1$ and the particle mass is $m=1/2$)
\begin{equation}\label{eqn.def.hamiltonian}
H_N = \sum_{j=1}^N -\Delta_{x_j} + \sum_{j < k} V(x_j - x_k)
\end{equation}
and is defined on (an appropriate domain in) $L^2_a(\Lambda^N; \C^q) := \bigwedge^N L^2(\Lambda; \C^q)$ 
with $q \in \N$ the number of spin states.
We consider such a system in the regime where the particle density $\rho = N/L^d$ is small compared to the length scale set by the interaction potential $V$.
The particles being fermions, an important role is played by the spin:
The interactions between fermions in different spin states gives a much larger contribution to the energy  than the interactions 
between  fermions in the same spin state. This is due to the Pauli exclusion principle suppressing the probability of 
two  fermions of equal spin being close. For fermions of different spin there is no such suppression.

In this paper we study the setting of  spinless (or, equivalently, fully spin-polarized) fermions, meaning that $q=1$. In the physics literature, one finds the following conjecture 
for the ground state energy $E_N = \mathrm{inf\, \spec\,} H_N$  \cite{Efimov.1966,Efimov.Amusya.1965,Amusia.Efimov.1968,Ding.Zhang.2019}
(see also the numerics in \cite{Bertaina.Tarallo.ea.2023}): for small $a k_F$
\begin{flalign}
E_N^{d=3} & = N k_F^2 \left[\frac{3}{5} 
	+ \frac{2}{5\pi} a^3 k_F^3  
	- \frac{1}{35\pi} a^6 R_{\textnormal{eff}}^{-1 }k_F^5
  	+ \frac{2066-312\log 2}{10395\pi^2} a^6 k_F^6 + o( a^6k_F^6) \right]
\label{eqn.conjectured}
\end{flalign}
with $k_F = (6\pi^2\rho)^{1/3}$ the Fermi momentum,
$a$ the $p$-\emph{wave scattering length} and $R_{\textnormal{eff}}$ the $p$-\emph{wave effective range}. 
The first term $\frac{3}{5} N k_F^2$ is the (kinetic) energy of the free (i.e., non-interacting) Fermi gas.
In \cite{Lauritsen.Seiringer.2024} the validity of the first three terms was proved as an upper bound. 
In this paper we shall prove the validity of the first two terms as a lower bound. In particular, 
in combination with the result in \cite{Lauritsen.Seiringer.2024} we establish the validity of \Cref{eqn.conjectured}
to order $Na^3k_F^5$. 

For comparison, let us consider the case of fermions with spin $\geq 1/2$, where there are $q\geq 2$  spin states.
To leading order, one only sees the interaction between fermions of different spins. In this case it is known  \cite{Falconi.Giacomelli.ea.2021,Giacomelli.2023,Lauritsen.2023,Lieb.Seiringer.ea.2005} that, 
for a gas with $N_\sigma = \rho_\sigma L^3$ fermions of spin $\sigma$, 
\begin{align}\label{E:q}
E_{ \{N_\sigma\} }^{d=3} 
	& = \sum_{\sigma} \frac{3}{5} N_\sigma (6\pi^2\rho_\sigma)^{2/3}
  + \sum_{\sigma \ne \sigma'} 4 \pi a_s N_\sigma N_{\sigma'} L^{-3}
  + o(N a_s \rho)
\end{align}
with $a_s$ the $s$-\emph{wave scattering length} of the interaction $V$. 
The first term $\sum_\sigma \frac{3}{5} N_\sigma (6\pi^2\rho_\sigma)^{2/3}$ is again 
the energy of a free Fermi gas.
We note that the second term, of order $Na_s \rho$, is much larger than the corresponding second term for the spin-polarized fermions in \Cref{eqn.conjectured}. 
The next term in the expansion \Cref{E:q} is conjectured to be of order $N a_s^2 \rho^{4/3}$ 
\cite{Huang.Yang.1957,Giacomelli.2023,Giacomelli.2023a}. 
Also this term arises  from  interactions of fermions with different spin and is still much larger 
than the largest term coming from the same-spin interaction in \Cref{eqn.conjectured} above.

\newcommand{\pushright}[1]{\ifmeasuring@#1\else\omit\hfill$\displaystyle#1$\fi\ignorespaces}

Finally, we consider also the lower-dimensional cases $d\leq 2$. Here the expected  formulas 
for the spin-polarized gas read \cite{Lauritsen.Seiringer.2024,Agerskov.Reuvers.ea.2022}
\begin{align*}
E_N^{d=2} 
& = N k_F^2 \left[\frac{1}{2} + \frac{1}{4} a^2 k_F^2 + o(a^2k_F^2)\right],
&
E_N^{d=1} 
& = N k_F^2 \left[\frac{1}{3} + \frac{2}{3\pi} ak_F + o(ak_F)\right]
\end{align*}
with $k_F = (4\pi \rho)^{1/2}$ for $d=2$ and $k_F = \pi\rho$ for $d=1$, and $a$ the $p$-\emph{wave scattering length} in the respective dimension. 
In \cite{Lauritsen.Seiringer.2024} we proved the validity of both of these formulas as upper bounds, and in 
\cite{Agerskov.Reuvers.ea.2022}  the one-dimensional formula is proved both as an upper and a lower bound. 
In this paper, we shall also prove the formula for $d=2$ as a lower bound.

\subsection{Precise statement of results}
We shall now give a precise statement of our main results, given in \Cref{thm.main} below. 
To do this, we first define the Hamiltonian $H_N$ and its ground state energy $E_N$ properly. 

We shall work with periodic boundary conditions on $\Lambda = [-L/2,L/2]^3$. 
In particular we replace the interaction $V$ by its periodization 
$\sum_{n \in \Z^3} V(x + nL)$, which we will with a slight abuse of notation continue to denote by $V$. 
We assume that $V$ has compact support, so for $L$ large enough at most one of the summands in $\sum_{n\in \Z^3} V(x+nL)$
is non-zero and no confusion should arise.
The Hamiltonian $H_N$ is then defined as in \Cref{eqn.def.hamiltonian} with $\Delta$ denoting the 
Laplacian with periodic boundary conditions on the box $\Lambda$
and realized as a self-adjoint operator on (an appropriate domain in) 
the fermionic space $L^2_a(\Lambda^N) = \bigwedge^N L^2(\Lambda)$.
The ground state energy $E_N$ is then given by 
\begin{equation*}
E_N = \inf_{\psi\in L^2_a(\Lambda^N)} \frac{\longip{\psi}{H_N}{\psi}}{\ip{\psi}{\psi}}.
\end{equation*}
The $p$-\emph{wave scattering length} of the interaction potential $V$ is defined as follows 
(a different-looking but equivalent definition is given in \cite{Lauritsen.Seiringer.2024,Lauritsen.Seiringer.2023a,Seiringer.Yngvason.2020}).
\begin{defn}\label{defn.scat.fun}
Let $\varphi_0$ be the solution of the  $p$-\emph{wave scattering equation}
\begin{equation}\label{eqn.scat}
x \Delta \varphi_0 + 2 \nabla \varphi_0 + \frac{1}{2} x V (1-\varphi_0) = 0
\end{equation}
on $\R^3$, 
with $\varphi_0(x) \to 0 $ for $|x|\to \infty$. Then $\varphi_0(x) = a^3/|x|^3$ for $x\notin \supp V$
for some constant $a$ called the $p$-\emph{wave scattering length}.
\end{defn}

With these definitions we can formulate our main theorem:
\begin{thm}\label{thm.main}
Let $V\in L^1$ 
be non-negative, radial and compactly supported. 
Then for $ak_F$ small enough and $N$ large enough we have 
\begin{equation*}
\frac{E_N}{N} \geq  k_F^2\left[\frac{3}{5} + \frac{2}{5\pi} a^3 k_F^3 + O( (ak_F)^{3+3/10} \abs{\log ak_F}) + O(N^{-1/3}) \right].
\end{equation*}
\end{thm}

\begin{remark}\label{rmk.length=a.in.bdds}
The appearance of the scattering length $a$ in  the error term is for dimensional consistency. The error term $O( (ak_F)^{3+3/10} \abs{\log ak_F})$
really depends on the range $R_0$ of $V$ and on $\norm{V}_{L^1}$ (both of dimension length). 
We think of $a$ as a constant of dimension length and thus use the bounds 
\begin{equation*}
R_0\leq C a, \qquad 
\norm{V}_{L^1}\leq Ca 
\end{equation*}
with the constants $C$ then being  dimensionless.
\end{remark}

\begin{remark}[{Extension to less regular $V$}]
A posteriori we can extend \Cref{thm.main} to less regular $V$
(and, in particular, to the case of hard spheres where, formally, $V_{\textnormal{hs}}(x) = \infty$ for $|x|\leq R_0$ and $V_{\textnormal{hs}}(x) = 0$ otherwise). 
Indeed, any positive radial and compactly supported measurable function $V$ can be approximated from below by some (positive radial compactly supported) 
$\widetilde V \in L^1$.
Then we can apply the theorem for $\widetilde V$ and note that $E_N \geq \widetilde E_N$ with $\widetilde E_N$ the ground state energy with interaction $\widetilde V$.
The error bounds in the theorem, being dependent on $\Vert\widetilde V\Vert_{L^1}$, 
necessarily blow up when $\widetilde V$ converges to $V$.
However, choosing $\widetilde V$ to converge to $V$ slowly enough we may achieve $a(\widetilde V) = a(V)(1+ o(1))$ with the error-terms of \Cref{thm.main} still being small.
Then 
\begin{equation*}
\frac{E_N}{N} \geq k_F^2 \left[\frac{3}{5} + \frac{2}{5\pi} a^3 k_F^3 + o(a^3k_F^3) + O(N^{-1/3})\right].
\end{equation*}
In the same way, also the restriction to $V$ having compact support can be lifted. For finiteness of the $p$-wave scattering length it is only needed that $x \mapsto |x|^2 V(x)$ is integrable outside some ball. 
\end{remark}

While the main focus of this paper is a {\em lower} bound on the ground state energy, our method can also be applied to obtain a corresponding {\em upper} bound. In fact, we shall show the following.

\begin{prop}[Upper bound]\label{prop.upper.bdd}
Let $V\in L^1$ 
be positive, radial and compactly supported. 
Then for $ak_F$ small enough and $N$ large enough we have 
\begin{equation*}
\frac{E_N}{N} \leq  k_F^2\left[\frac{3}{5} + \frac{2}{5\pi} a^3 k_F^3 + O( a^4 k_F^4 \abs{\log ak_F}) + O(N^{-1/3}) \right].
\end{equation*}
\end{prop}
We remark, however, that in \cite[Theorem 1.3]{Lauritsen.Seiringer.2024}, using a very different method,  a significantly stronger upper bound was shown (capturing also the  next term of order $(a k_F)^5$) under weaker assumptions on the interaction $V$ (in particular, allowing also for hard spheres).

%\begin{remark}
In the proof of \Cref{thm.main} and \Cref{prop.upper.bdd} we will consider particle numbers $N$ arising from a ``filled Fermi ball''.
This is done for convenience. 
We shall discuss  in \Cref{rem.choice.N=B_F} below why this is in fact not a restriction 
on $N$ to the precision given in \Cref{thm.main}.
Concretely this means the following.
\begin{defn}
For $k_F>0$ the \emph{Fermi ball} is defined as $B_F = \{k\in \frac{2\pi}{L}\Z^3 : |k|\leq k_F\}$.
We then take $N = \# B_F$.
\end{defn}
%\end{remark}

%\begin{remark}
In this case there are two variables, which we are free to choose:
The side-length of the box $L$ and the Fermi momentum $k_F$.
The particle number $N$ and density $\rho = N/L^3$ then depend on the values of $L, k_F$.
Any constraint on $N$ (that it is sufficiently large, say) 
should thus more precisely be written as a constraint on $k_FL \sim N^{1/3}$. 
%\end{remark}
%\begin{remark}
Counting the number of lattice points inside a ball of a given radius we see that $\rho = N/L^3$ and $k_F$ are related by 
$k_F = (6\pi^2\rho)^{1/3}(1 + O(N^{-1/3}))$.
%\end{remark}

\begin{remark}\label{rem.choice.N=B_F}
The choice of $N = \# B_F$ puts a restriction on which values $N$ may assume ---
not all integers arise as $\# B_F$ for some $L$ and $k_F$.
To the precision given in \Cref{thm.main} and \Cref{prop.upper.bdd}, however, 
it suffices to consider $N$ arising as $N=\# B_F$. 
To see this, assume that \Cref{thm.main} and \Cref{prop.upper.bdd} hold for integers $N$ arising as $N=\# B_F$.

For a general integer $N$ and length $L$ define $k_F^<$ and $k_F^>$ as the largest, 
respectively smallest, $k_F$ such that 
$N^< = \# B_F^< \leq N \leq \# B_F^> = N^>$ with $B_F^<, B_F^>$ defined using the Fermi momenta $k_F^<$ and $k_F^>$.
Then $N^<, N^> = N + O(N^{2/3})$ since $k_F^< + \frac{2\pi}{L} \geq k_F^>$.
Moreover, $k_F^<,k_F^> = (6\pi^2\rho)^{1/3}(1+ O(N^{-1/3}))$.
We may apply \Cref{thm.main} and \Cref{prop.upper.bdd} for particle numbers $N^>$ and $N^<$. % (assuming they are sufficiently large).
By positivity of the interaction we have $E_{N^<} \leq E_N \leq E_{N^>}$.
Thus,
\begin{equation*}
\begin{aligned}
  \frac{E_N}{N}
    \geq \frac{E_{N^<}}{N^<} ( 1 + O(N^{-1/3}))
    & 
    \geq (k_F^<)^2\left[\frac{3}{5} + \frac{2}{5\pi} a^3 (k_F^<)^3 + O( (ak_F^<)^{3+3/10} \abs{\log ak_F^<}) + O(N^{-1/3}) \right]
    \\
    & = k_F^2\left[\frac{3}{5} + \frac{2}{5\pi} a^3 k_F^3 + O( (ak_F)^{3+3/10} \abs{\log ak_F}) + O(N^{-1/3}) \right]
\end{aligned}
\end{equation*}
and similarly for the upper bound. 
\end{remark}

\subsubsection{Two dimensions}
We consider next the analogous problem in two dimensions. 
The scattering length is defined as in \Cref{defn.scat.fun}, with the only difference that now 
 one has $\varphi_0(x) = a^2/|x|^2$ outside the support of $V$.
The two-dimensional analogue of \Cref{thm.main} is as follows.

\begin{thm}[Two dimensions]
\label{thm.main.d=2}
Let $V\in L^1$ 
be positive, radial and compactly supported. 
Then for $ak_F$ small enough and $N$ large enough we have 
\begin{equation*}
\frac{E_N}{N} \geq  k_F^2\left[\frac{1}{2} + \frac{1}{4} a^2 k_F^2 + O( (ak_F)^{2+1/4} \abs{\log ak_F}) + O(N^{-1/2}) \right].
\end{equation*}
\end{thm}

\begin{remark}
\Cref{thm.main.d=2} matches the upper bound of \cite[Theorem 1.10]{Lauritsen.Seiringer.2024} to order $N a^2 k_F^4$. 
Indeed, $k_F$ and $\rho$ are related by $k_F = (4\pi \rho)^{1/2}(1 + O(N^{-1/2}))$.
\end{remark}

We sketch in \Cref{sec.d=2} how to adapt the proof of \Cref{thm.main} to the two-dimensional setting. 

\subsection{Second quantization}
An important step in the proof is to write the Hamiltonian in second quantization 
and analyse it using an appropriate unitary transformation. The choice of the unitary can be motivated by 
a suitable version of second order perturbation theory, as we shall discuss in the next subsection.

We will here only briefly describe the central concepts of second quantization.
A detailed introduction to second quantization in general can be found in \cite{solovej.qm.lecture}.
Moreover, the specific case of a Fermi gas in a periodic box $\Lambda = [-L/2,L/2]^3$ is discussed in detail in \cite{Falconi.Giacomelli.ea.2021,Giacomelli.2023},
see also \cite{Giacomelli.2023a}.

Since we consider the box $\Lambda$ with periodic boundary conditions, a natural basis for the one-body space $L^2(\Lambda)$ is  
given by the plane waves $f_k(x) = L^{-3/2} e^{ikx}$ with momenta $k\in \frac{2\pi}{L}\Z^3$.
We will denote the creation and annihilation operators in the state $f_k$ by $a_k^* = a^*(f_k)$ and $a_k = a(f_k)$, respectively.
They satisfy the canonical anti-commutation relations 
$\{a_k^*, a_{k'}\} = \delta_{k,k'}$ and $\{a_k, a_{k'}\} = \{a_k^*, a_{k'}^*\} = 0$.
In second quantization the Hamiltonian is then given by 
\begin{equation}\label{eqn.def.H.2nd.quant}
\mcH 
= \ud \Gamma(-\Delta) + \ud \Gamma(V)
= \sum_k |k|^2 a_k^* a_k + \frac{1}{2L^3} \sum_{p,k,k'} \hat V(p) a_{k+p}^* a_{k'-p}^* a_{k'} a_k
\end{equation}
and defined on the Fock space $\mcF = \bigoplus_{n=0}^\infty L^2_a(\Lambda^n) = \bigoplus_{n=0}^\infty \bigwedge^n L^2(\Lambda)$.
Here we have adopted the notation
\begin{notation}
The Fourier transform of a function $g$ (on $\Lambda$) is given by $\hat g(k) = \int g(x) e^{-ikx} \ud x$.
\end{notation}

\begin{notation}
For any sum the variables are summed over $\frac{2\pi}{L}\Z^3$ unless otherwise noted. 
That is, $\sum_k = \sum_{k\in \frac{2\pi}{L}\Z^3}$.
\end{notation}  
We will consider $N$-particle states $\psi$, meaning that $\mcN \psi = N \psi$ with 
$\mcN = \sum_k a_k^* a_k$ the number operator.  To extract the leading contribution to the ground state energy it is convenient to introduce the particle-hole transformation $R$, 
satisfying
\begin{equation}\label{def:ph}
R^* a_k R = 
\begin{cases}
a_k^* & k \in B_F, \\ a_k & k \notin B_F,
\end{cases}
\qquad 
R \Omega = \left[\prod_{k\in B_F} a_k^*\right] \Omega = \psi_F.
\end{equation}
Here $\Omega$ denotes the vacuum and the free Fermi state $\psi_F$ is given by 
\begin{equation*}
\psi_F(x_1,\ldots,x_N) = \frac{1}{\sqrt{N!}} \det \left[f_k(x_j)\right]_{\substack{k \in B_F \\ 1\leq j \leq N}},
\end{equation*}
with $N = \# B_F$. % and $f_k(x) = L^{-3/2} e^{ikx}$ the plane waves as above. 
The state $\psi_F$ is the ground state of the corresponding non-interacting system. 
For later use we introduce the following convenient notation. 

\begin{notation}
The expectation of an operator $A$ in the free Fermi state $\psi_F$ is denoted by $\expect{A}_F$.
\end{notation}

It will prove helpful to distinguish between creation and annihilation operators 
for momenta inside and outside the Fermi ball.  
Moreover, it will sometimes be convenient  to consider operators written in configuration space, 
i.e., using the operator-valued distributions $a_x, a_x^*$ given by $a_x = L^{-3/2}\sum_k e^{ikx} a_k$. 
For instance we have 
\begin{equation*}
\ud \Gamma(V) = \frac{1}{2} \iint V(x-y) a_x^* a_y^* a_y a_x \ud x \ud y.
\end{equation*}
Also for the operator-valued distributions $a_x,a_x^*$ we wish to 
be able to distinguish whether they arise from particles inside or outside the Fermi ball. 
This leads to the following definition.

\begin{defn}
Define
\begin{equation*}
c_k = a_k \chi_{(k\in B_F)},
\qquad 
b_k = a_k \chi_{(k\notin B_F)},
\end{equation*}
with $\chi$ denoting the characteristic function.
Define further the operator-valued distributions 
\begin{equation}
b_x = \frac{1}{L^{3/2}} \sum_{k} e^{ikx} b_k = \frac{1}{L^{3/2}} \sum_{k\notin B_F} e^{ikx} a_k,
\qquad 
c_x = \frac{1}{L^{3/2}} \sum_{k} e^{-ikx} c_k = \frac{1}{L^{3/2}} \sum_{k\in B_F} e^{-ikx} a_k.
\label{eqn.def.bx.cx}
\end{equation}
\end{defn}
Note the different choice of signs in the exponents. This is done for convenience, so that the particle-hole transformation satisfies
\begin{equation}\label{eqn.prop.particle-hole.R}
R^* a_x R = b_x + c_x^*.
\end{equation}
Note that $c_x$ and $c_x^*$ are in fact bounded operators with $\norm{c_x}=\norm{c_x^*}\leq C k_F^{3/2}$.
Moreover, since their supports (in momentum-space) are disjoint we have that $b_x$ and $c_y^*$ anti-commute.

For later use we define the operators $u,v$ as follows.
\begin{defn}\label{def.u.v.proj}
Define the operators $u,v$ as the projection outside and inside the Fermi ball, respectively. 
That is, their kernels are given by (with a slight abuse of notation)
\begin{equation*}
\begin{aligned}
v(x;y) & = v(x-y) = \frac{1}{L^3}\sum_k \hat v(k) e^{ikx} = \frac{1}{L^3} \sum_{k \in B_F} e^{ikx},
\\
u(x;y) & = u(x-y) = \frac{1}{L^3}\sum_k \hat u(k) e^{ikx} = \frac{1}{L^3} \sum_{k \notin B_F} e^{ikx} = \delta(x-y) - v(x-y).
\end{aligned}
\end{equation*}
\end{defn}

\begin{remark}
The operator-valued distributions $b_x$ and $c_x$ are denoted $a(u_x) = a(u(\cdot;x))$ and $a(\overline{v}_x) = a(\overline{v}(\cdot;x))$, 
respectively, in \cite{Falconi.Giacomelli.ea.2021,Giacomelli.2023}.
The different signs in the exponents in \Cref{eqn.def.bx.cx} above reflect the $\overline{\phantom{v}}$ in $a(\overline{v}_x)$.
\end{remark}

\subsection{Heuristics}
Conceptually, the strategy of the proof of \Cref{thm.main} can be motivated via (a suitable version of) second order perturbation theory, 
where one treats only the ``off-diagonal'' parts of the interaction 
as the perturbation and includes the ``diagonal'' part in the unperturbed operator. This is similar to the method of \cite{Giacomelli.2023,Falconi.Giacomelli.ea.2021}, which itself is inspired by Bogoliubov transformations in the case of bosons; the latter effectively appear as pairs of fermions with opposite spin. This (bosonic) picture does not apply in the spinless case, but the method can be applied nonetheless, as we shall show below.

Before we explain this is more details we first recall second order perturbation theory in general.

\subsubsection{General second order perturbation theory}
Consider a generic perturbed Hamiltonian $H = H_0 + \lambda V$ and write $V = V_{\textnormal{D}} + V_{\textnormal{OD}}$
with $V_{\textnormal{D}}$ the part of $V$ diagonal in a basis where $H_0$ is diagonal. 
We use a formulation of second order perturbation theory using a unitary operator $e^{\lambda  B}$ with $ B$ chosen appropriately
such that $e^{-\lambda  B}H e^{\lambda  B}$ is (approximately) diagonal. 
Using the Baker--Campbell--Hausdorff formula
$e^{-X}Y e^X \approx Y + [Y,X] + \frac{1}{2}[ [Y, X], X] + \ldots$
we have 
\begin{equation*}
\begin{aligned}
e^{-\lambda  B}H e^{\lambda  B}
  & = H_0 + \lambda \left([H_0, B] + V\right) + \lambda^2 \left(\frac{1}{2} [ [H_0,  B],  B] + [V, B]\right)  
    + O(\lambda^3)
  \\
  & = H_0 + \lambda \left([H_0 + \lambda V_{\textnormal{D}}, B] + V\right) 
    + \lambda^2 \left(\frac{1}{2} [ [H_0+ \lambda V_{\textnormal{D}},  B],  B] + [V_{\textnormal{OD}}, B]\right)  
    + O(\lambda^3)
\end{aligned}
\end{equation*}
Since $H_0$ and $V_{\textnormal{D}}$ are both diagonal operators, $[H_0 + \lambda V_{\textnormal{D}}, B]$ is off-diagonal for any $ B$.
We  choose $ B$ to (approximately) cancel the off-diagonal part of $V$, i.e., such that $[H_0 + \lambda V_{\textnormal{D}}, B] \approx - V_{\textnormal{OD}}$.
Then 
\begin{equation}\label{eqn.2nd.order.perturbation}
e^{-\lambda  B}H e^{\lambda  B} 
  \approx H_0 + \lambda V_{\textnormal{D}} + \frac{1}{2} \lambda^2 [V_{\textnormal{OD}}, B]
\end{equation}
is  diagonal to order $\lambda$. 
With $e_n^{(0)} = \longip{n}{H_0 + \lambda V_{\textnormal{D}}}{n}$ the eigenvalues of $H_0 + \lambda V_{\textnormal{D}}$ the equation for $B$ reads 
$\longip{n}{B}{m}(e_n^{(0)}  -  e_{m}^{(0)} )= \longip{n}{V}{m}$.
Thus, we find the well-known formula for the  eigenvalues $e_n$ of $H$
as the diagonal matrix-elements of $e^{-\lambda  B}H e^{\lambda  B}$ 
\begin{equation*}
e_n \approx e_n^{(0)} + \lambda^2 \sum_{m:m\ne n} \frac{\abs{\longip{n}{V}{m}}^2}{e_n^{(0)} - e_{m}^{(0)}}
\approx
\longip{n}{H_0}{n} + \lambda \longip{n}{V}{n} + \lambda^2 \sum_{m:m\ne n} \frac{\abs{\longip{n}{V}{m}}^2}{\longip{n}{H_0}{n} - \longip{m}{H_0}{m}}
\end{equation*}
valid to order $\lambda^2$.

If one is only interested in the ground state energy, one can simplify the computation above. 
Decompose the Hilbert space as $\tspan \{\ket{0}\}\oplus\tspan \{\ket{0}\}^\perp$ with $\ket{0}$ being the ground state of $H_0$.
We then instead take $V_\textnormal{D}$ as the part of $V$ \emph{block} diagonal in this decomposition and 
choose $B$ to only cancel the \emph{block} off-diagonal part of $V$.
(This is  in general a considerably simpler choice of $B$.)
The formula in \Cref{eqn.2nd.order.perturbation} then still applies, being however only (approximately) \emph{block} diagonal,
and we find the ground state energy as 
\begin{equation}\label{2op}
e_0 \approx e_0^{(0)} + \lambda^2 \sum_{m\ne 0} \frac{\abs{\longip{0}{V}{m}}^2}{e_0^{(0)} - e_{m}^{(0)}}.
\end{equation}

In our case of interest, there is no small coupling parameter $\lambda$. Nevertheless, the relevance of \eqref{2op} is suggested by the following observation. If one integrates \Cref{eqn.scat} against $x$, one finds via an integration by parts that 
\begin{equation}\label{a3}
12 \pi a^3 = \frac 12 \int V(x) |x|^2 (1-\varphi_0(x))\ud x .
\end{equation}
\Cref{eqn.scat} can be rewritten as $(-x\Delta -2 \nabla + \frac 12 x V)\varphi_0 = \frac 12 xV$, hence \eqref{a3} is reminiscent of \eqref{2op}, only that it is actually exact!

\subsubsection{Application to the many-body setting}
Our goal is to apply the formula in \Cref{eqn.2nd.order.perturbation} to the Hamiltonian $H_N$,  with $\lambda = 1$. 
This is not a small parameter ---  the potential $V$ is not assumed to be small. 
For \Cref{eqn.2nd.order.perturbation} to still be a good approximation, it is essential that the ``diagonal'' part of $V$ 
is taken as part of the unperturbed Hamiltonian.
Indeed, the formula in \Cref{thm.main} is non-perturbative and \emph{does not} arise from any finite order ``standard'' perturbation theory,
where all of $V$ is taken as the perturbation.

First, we need to define what we mean by ``diagonal'' and ``off-diagonal''.
Recall from \Cref{def.u.v.proj} that $v$ and $u=\mathbbm{1}-v$ are the projections inside, respectively outside, the Fermi ball. 
For the two-body operator $V$ we then  have 
\begin{equation*}
V \approx V_{\textnormal{D}} + V_{\textnormal{OD}},
\qquad 
V_{\textnormal{D}} = vvVvv + uuVuu,
\qquad 
 V_{\textnormal{OD}} = vvVuu + uuVvv.
\end{equation*}
(By an expression like $vv$ we mean the two-body operator $v\otimes v$.)
Here we have neglected terms with a factor $uv$ or $vu$. These will turn out to be small. 
The ground state of the free Hamiltonian $\sum_{j=1}^N -\Delta_{x_j}$ is the free Fermi state $\psi_F$. 
Hence the calculation above suggests that the ground state of $H_N$ is roughly $e^{\widetilde B}\psi_F$
with $\widetilde B$ the analogue of the operator $B$ from above.

To compute $\longip{e^{\widetilde B} \psi_F}{H_N}{e^{\widetilde B} \psi_F}$
it is convenient to first conjugate the Hamiltonian $H_N$ by the particle-hole transformation $R$, defined in \Cref{def:ph} above. 
Define also $B = R^* \widetilde B R$.
Then the ground state of $H_N$ is expected to be roughly $e^{\widetilde B}\psi_F = R T \Omega$ with $T=e^B$.
Our choice of $B$ will be of the form 
\begin{equation*}
B = \frac{1}{2L^3} \sum_{p,k,k'} \hat\varphi(p) b_{k+p}b_{k'-p}c_{k'} c_k - \hc
  = \frac{1}{2}\iint \varphi(z-z') b_z b_{z'} c_{z'} c_z \ud z \ud z' - \hc
\end{equation*}
for $\varphi \approx \varphi_0$. 
This particular choice of $B$ will be justified by the validity of \Cref{eqn.aprrox.scat.eqn.H0+Q} below.

\begin{remark}
The form of the operator $B$ is (up to the spin dependence) 
the same as for the operator $B - B^*$ considered in \cite{Falconi.Giacomelli.ea.2021,Giacomelli.2023}.
Furthermore, we note that $\widetilde B$ is closely related to the \emph{transfer operator} $M_2$ considered in \cite{Amusia.Efimov.1968}.
In \cite{Amusia.Efimov.1968} it is claimed that the ground state is (up to normalization) approximately given by $(1+M_2)\psi_F$. 
Pretending that $\widetilde B \sim M_2$ is small we find that the ground state is roughly $e^{\widetilde B} \psi_F = RT\Omega$ as claimed.
\end{remark}

To compute the expectation of $H_N$ in the state $\psi \approx RT\Omega$ 
we first compute the conjugation of the Hamiltonian $\mcH$ by the particle-hole transformation $R$.
(Recall that $\mcH$, defined in \Cref{eqn.def.H.2nd.quant}, is the second quantized analogue of $H_N$.) 
To calculate $R^*\mcH R$ we use \Cref{eqn.prop.particle-hole.R} and normal order. 
This is a straightforward calculation which we omit. 
The details of the calculation can for instance be found in \cite[Proposition 4.1]{Hainzl.Porta.ea.2020}.
We have for the kinetic energy 
(on the space of states of the form $R^*\psi$ with $\psi$ an $N$-particle state, see \cite[Proposition 4.1]{Hainzl.Porta.ea.2020})
\begin{equation}\label{eqn.KE.particle-hole}
R^*\ud\Gamma(-\Delta) R = E_F + \H_0,
\qquad 
E_F = \sum_{k\in B_F} |k|^2,
\qquad 
\H_0 = \sum_k \abs{|k|^2 - k_F^2} a_k^* a_k
\end{equation}
and for the interaction 
(in the second $\approx$ by neglecting the quadratic terms)
\begin{equation*}
R^* \!\ud \Gamma(V) R 
	\approx R^* \!\ud \Gamma (V_{\textnormal{D}} + V_{\textnormal{OD}}) R
  \approx \expect{\!\ud \Gamma (V)}_F + \Q_2 + \Q_4
\end{equation*}
with 
$\Q_2$ and $\Q_4$ given by 
\begin{align}
\Q_2
  & = \frac{1}{2L^3} \sum_{p,k,k'} \hat V(p) b_{k+p}^* b_{k'-p}^* c_{k'}^* c_k^* + \hc 
  \hspace*{-2em}
  &&  = \frac{1}{2} \iint V(x-y) b_x^* b_y^* c_y^* c_x^* \ud x \ud y + \hc ,
  \label{eqn.def.Q2}
  \\
\Q_4 
  & = \frac{1}{2L^3} \sum_{p,k,k'} \hat V(p) b_{k+p}^* b_{k'-p}^* b_{k'} b_k
  \hspace*{-2em}
  &&  = \frac{1}{2} \iint V(x-y) b_x^* b_y^* b_y b_x \ud x \ud y .
  \label{eqn.def.Q4}
\end{align}
Here 
$\expect{\!\ud\Gamma(V)}_F + \Q_4$ is the operator corresponding to the ``diagonal'' part $vvVvv + uuVuu$
and 
$\Q_2$ is the operator corresponding to the ``off-diagonal'' part $vvVuu + uuVvv$.

\begin{remark}[{Comparison to \cite{Falconi.Giacomelli.ea.2021,Giacomelli.2023}}]
The analogues of the operators $\Q_2$ and $\Q_4$ above are denoted by $\Q_4$ and $\Q_1$, respectively, in \cite{Falconi.Giacomelli.ea.2021,Giacomelli.2023}.
\end{remark}

The calculation of general second order perturbation theory above then translates to the calculation 
\begin{equation}\label{eqn.BCH.heuristic}
\begin{aligned}
T^*(\H_0 + \Q_2 + \Q_4)T
 & \approx 
  \H_0 + \Q_4
  + [\H_0 + \Q_4,B] + \Q_2
  + \frac{1}{2}[[\H_0 + \Q_4, B], B] + [\Q_2, B].
\end{aligned}
\end{equation}
The operator $B$ is chosen such that 
\begin{equation}\label{eqn.aprrox.scat.eqn.H0+Q}
[\H_0 + \Q_4,B] + \Q_2 \approx 0,	
\end{equation}
which constrains the function $\varphi$ in the definition of $B$ to (roughly) satisfy the $p$-wave scattering equation \eqref{eqn.scat}. 
Then,
\begin{equation*}
\begin{aligned}
T^*(\H_0 + \Q_2 + \Q_4)T
  & \approx 
    \H_0 + \Q_4
  + \frac{1}{2}[\Q_2, B].
\end{aligned}
\end{equation*}
Next, one computes that 
$\expect{[\Q_2, B]}_\Omega \approx - 2\expect{\!\ud\Gamma(V\varphi)}_F$.
We then find 
\begin{equation*}
E_N \approx E_F + \expect{\H_0 + \Q_4}_\Omega + \expect{\!\ud\Gamma(V)}_F + \frac{1}{2}\expect{[\Q_2, B]}_\Omega
  \approx E_F + \expect{\!\ud\Gamma(V(1-\varphi))}_F.
\end{equation*}
Finally, we compute that 
$E_F \approx \frac{3}{5}N k_F^2$ and  
$\expect{\!\ud\Gamma(V(1-\varphi))}_F \approx \frac{2}{5\pi}N a^3 k_F^5$ (see \Cref{eqn.calc.<Vphi>.F} below).
This recovers the formula of \Cref{thm.main}.

\section{Rigorous Analysis}
In this section we shall describe how to rigorously implement the heuristic computation above in order to prove \Cref{thm.main}. 
A first step is the precise definition of the operators $B$ and $T$.

\subsection{Regularizing the operators} For technical reasons, we need to regularize the operator(-valued distribution)s $b_k$ and $b_x$.
We introduce  regularized $b_k$-operators using a smooth radial function $\hat u^r : \R^3 \to [0,1]$ with 
\begin{equation*}
\hat u^r(k) = \begin{cases}
0 & |k|\leq 2 k_F,
\\
1 & |k|\geq 3 k_F.
\end{cases}
\end{equation*}
Let 
\begin{equation}\label{def:br}
b_k^r = \hat u^r(k) a_k, \qquad b_x^r = \frac{1}{L^{3/2}} \sum_k e^{ikx} b_k^r.
\end{equation}
The function $\hat u^r$ defines an operator $u^r$ on $L^2(\Lambda)$ 
with kernel (with slight abuse of notation) $u^r(x;y) = u^r(x-y)$.
As an operator $0 \leq u^r \leq \mathbbm{1}$. Moreover, by proceeding as in \cite[Proposition 4.2]{Falconi.Giacomelli.ea.2021} one easily checks that $u^r - \delta \in L^1$. More precisely:

\begin{lemma}
\label{lem.bdd.vr.L1}
Write $u^r = \delta - v^r$. Then $\norm{v^r}_{L^1}\leq C$ for large $k_F L$.
\end{lemma}

\subsection{The scattering function}
We shall now define the function $\varphi$ appearing in the operator $B$.
Since $\varphi_0$ is not integrable at infinity (it decays like $a^3/|x|^3$)
it  will be necessary to introduce a cut-off for technical reasons. 
Moreover, we need to periodize to define $\varphi$ on the torus.

Let $\chi_\varphi: [0,\infty) \to [0,1]$ be a smooth function with 
\begin{equation*}
\chi_\varphi (t) = \begin{cases}
1 & t \leq 1, 
\\
0 & t \geq 2.
\end{cases}
\end{equation*}
Then %(with a slight abuse of notation) 
we choose the function $\varphi$ as 
\begin{equation}\label{eqn.def.scat.fun.phi}
\varphi(x) = \varphi_0(x) \chi_\varphi(k_F|x|)
\end{equation}
with $\varphi_0$ the  $p$-wave scattering function defined in \Cref{defn.scat.fun}.
To be more precise, we choose $\varphi$ to be the periodization of $\varphi_0 \chi_\varphi(k_F|\cdot|)$. 
We shall always assume that $L > 4k_F^{-1}$ so that $\supp \varphi_0 \chi_\varphi(k_F|\cdot|) \subset \Lambda$,
in which case no confusion should arise.

To measure the error in $\varphi$ not exactly satisfying the scattering equation \eqref{eqn.scat} we define 
\begin{equation}\label{eqn.defn.mcE_varphi.scat.eqn}
\begin{aligned}
\mcE_\varphi(x) 
  & = x \Delta \varphi(x) + 2 \nabla \varphi(x) + \frac{1}{2} x V(x) (1-\varphi(x))
  \\ & = 2 k_F x \nabla^\mu \varphi_0(x) \nabla^\mu \chi_\varphi(k_Fx)
  + k_F^2 x \varphi_0(x) \Delta\chi_\varphi(k_Fx)
  + 2 k_F \varphi_0(x) \nabla \chi_\varphi(k_Fx).
\end{aligned}
\end{equation}
Again, $\mcE_\varphi$ is more precisely the periodization of the expression in the second line. 
Here we have used the notation:
\begin{notation}\label{Enot}
We adopt the Einstein summation convention of summing over repeated indices denoting components of a vector,
i.e., $x^\mu y^\mu = \sum_{\mu=1}^3 x^\mu y^\mu = x\cdot y$ for vectors $x,y\in \R^3$.
\end{notation}
\begin{remark}
We will in general abuse notation slightly and refer to any (compactly supported) function and its periodization by the same name. 
For sufficiently large $L$ at most one summand in the periodization is non-zero and so no issue will arise. 
\end{remark}

\begin{remark}
Any derivative of $\chi_\varphi(t)$ is supported in $1\leq t\leq 2$. In particular $\mcE_\varphi$ is supported in the annulus $k_F^{-1} \leq |x|\leq 2k_F^{-1}$.
For $ak_F$ small enough such $x$ are outside the support of $V$ and hence $\varphi_0(x) = a^3/|x|^3$ there.  
In particular 
\begin{equation}\label{eqn.bdd.mcE_varphi}
\abs{\mcE_\varphi(x)} \leq C k_F \frac{a^3}{|x|^3} \chi_{k_F^{-1}\leq |x|\leq 2 k_F^{-1}} \leq C k_F \varphi(x/2).
\end{equation}
\end{remark}

\subsection{The transformation \texorpdfstring{$T$}{T}}
The operators $B$ and $T$ are defined as 
\begin{equation}\label{eqn.def.B}
T = e^{B}, 
\qquad 
B = \frac{1}{2L^3} \sum_{p,k,k'} \hat\varphi(p) b_{k+p}^r b_{k'-p}^r c_{k'} c_k - \hc
  = \frac{1}{2}\iint \varphi(z-z') b_z^r b_{z'}^r c_{z'} c_z \ud z \ud z' - \hc.
\end{equation}
with $\varphi$ given in \Cref{eqn.def.scat.fun.phi} and $b_k^r$ the regularized operators in \eqref{def:br}.

\begin{remark}[{Comparison to \cite{Falconi.Giacomelli.ea.2021,Giacomelli.2023}}]
Compared to the analogous construction of an operator $T$ in \cite{Falconi.Giacomelli.ea.2021,Giacomelli.2023}
we avoid regularizing the operators $c_x$ and we do not have an ultraviolet regularization of the operators $b_x$. 
This simplifies much of the analysis below, since we do not have to treat the errors arising from such a regularization. 
\end{remark}

\subsection{Implementation}
We shall describe now how to rigorously implement the heuristic calculation in the previous section. 
The main step is the calculation in \Cref{eqn.BCH.heuristic}, which uses the Baker--Campbell--Hausdorff expansion. 
We give a rigorous implementation of this calculation using a Duhamel expansion as follows.

As discussed in the heuristic analysis above, we expect the ground state to be roughly $RT\Omega$. 
Thus define (for some $\psi$)
\begin{equation}\label{eqn.def.xi.lambda}
\xi_\lambda = T^{-\lambda}R^* \psi  \ , \quad  \text{where} \  T^{-\lambda}:= e^{-\lambda B} , \quad \lambda \in [0,1].
\end{equation}
Then $\psi = R\xi_0 = RT\xi_1$.
We will consider the following $\psi$'s:
\begin{defn}
Let $\psi$ be an $N$-particle state. Then $\psi$ is called an \emph{approximate ground state} 
if 
\begin{equation*}
\abs{\expect{H_N}_\psi - \sum_{k\in B_F} |k|^2} \leq C N a^3 k_F^5.
\end{equation*}
for some $C>0$ (independent of $N$ and $a k_F$).
\end{defn}

\begin{remark}[Existence of approximate ground states]
Note that the free Fermi state $\psi_F$ is an approximate ground state. 
Indeed, we show in the proof of \Cref{lem.ini.V=Q2+Q4} below that 
$\expect{\!\ud\Gamma(V)}_F \leq C N a^3 k_F^5$.
In particular, the ground state is an approximate ground state. 
\end{remark}

Motivated by the discussion above we write 
\begin{equation}\label{eqn.main.expand.V=Q2+Q4}
\expect{H_N}_\psi = E_F + \expect{\!\ud\Gamma(V)}_F + \expect{\H_0 + \Q_2 + \Q_4}_{\xi_0} + \mcE_V(\psi)
\end{equation}
with $\mcE_V(\psi)$ defined so that this holds, i.e., $\mcE_V(\psi) = \expect{\!\ud\Gamma(V)}_\psi - \expect{\!\ud\Gamma(V)}_F - \expect{\Q_2 + \Q_4}_{\xi_0}$.
We shall show in \Cref{prop.mcE_V} below that $\mcE_V(\psi)$ is small for an approximate ground state $\psi$, i.e., it is negligible compared to $N a^3 k_F^5$. 
To compute $\expect{\H_0 + \Q_2 + \Q_4}_{\xi_0}$ 
we note that for any operator $A$ we have 
\begin{equation*}
\begin{aligned}
\longip{\xi_\lambda}{A}{\xi_\lambda}
  & = \longip{\xi_1}{A}{\xi_1} -\int_\lambda^1 \ud \lambda'\, \dd{\lambda'} \longip{\xi_{\lambda'}}{A}{\xi_{\lambda'}}
  = \longip{\xi_1}{A}{\xi_1} + \int_\lambda^1 \ud \lambda'\, \longip{\xi_{\lambda'}}{[A,B]}{\xi_{\lambda'}}.
\end{aligned}
\end{equation*}
Using this for $\H_0$ and $\Q_4$ we find 
\begin{equation*}
\begin{aligned}
\expect{\H_0 + \Q_4}_{\xi_0}
  & = \longip{\xi_1}{\H_0 + \Q_4}{\xi_1} + \int_0^1 \ud \lambda\, \longip{\xi_\lambda}{[\H_0 + \Q_4,B] + \Q_2}{\xi_\lambda}
    - \int_0^1 \ud \lambda \, \longip{\xi_\lambda }{\Q_2}{\xi_\lambda}.
\end{aligned}
\end{equation*}
Similarly, for $\Q_2$ we find 
\begin{equation*}
\begin{aligned}
\expect{\Q_2}_{\xi_0} 
  & = \longip{\xi_1}{\Q_2}{\xi_1} + \int_0^1 \ud \lambda\, \longip{\xi_{\lambda}}{[\Q_2,B]}{\xi_{\lambda}},
\\
\int_0^1 \ud \lambda \, \longip{\xi_\lambda }{\Q_2}{\xi_\lambda}
& = 
\longip{\xi_1}{\Q_2}{\xi_1} +
\int_0^1 \ud \lambda\int_\lambda^1 \ud \lambda'  \longip{\xi_{\lambda'}}{[\Q_2,B]}{\xi_{\lambda'}}.
\end{aligned}
\end{equation*}
In conclusion then 
\begin{equation*}
\begin{aligned}
\longip{\xi_0}{\H_0 + \Q_2 + \Q_4}{\xi_0}
  & = 
  \int_0^1 \ud \lambda\, \longip{\xi_{\lambda}}{[\Q_2,B]}{\xi_{\lambda}}
  - \int_0^1 \ud \lambda\int_\lambda^1 \ud \lambda'  \longip{\xi_{\lambda'}}{[\Q_2,B]}{\xi_{\lambda'}}
\\ & \quad + 
\longip{\xi_1}{\H_0 + \Q_4}{\xi_1} + \int_0^1 \ud \lambda\, \longip{\xi_\lambda}{[\H_0 + \Q_4,B] + \Q_2}{\xi_\lambda} .
\end{aligned}
\end{equation*}
The two terms with $[\Q_2,B]$ are the leading terms. 
To leading order they are given by the constant (i.e. fully contracted) contribution from $[\Q_2,B]$ obtained after normal ordering. 
This constant term is (approximately) given by $-2\expect{\!\ud\Gamma(V\varphi)}_F$. 
This motivates to write 
\begin{equation}\label{comQ2B}
[\Q_2,B] = -2\expect{\!\ud\Gamma(V\varphi)}_F + \Q_{2;B}^\mcE.
\end{equation}
The last term, with $[\H_0 + \Q_4,B] + \Q_2$, is an error term. 
It is small by virtue of $\varphi$ (almost) satisfying the $p$-wave scattering equation.
The third term $\longip{\xi_1}{\H_0 + \Q_4}{\xi_1}$ is positive.
We thus find 
\begin{equation*}
\begin{aligned}
\longip{\xi_0}{\H_0 + \Q_2 + \Q_4}{\xi_0}
  & = - \expect{\!\ud\Gamma(V\varphi)}_F 
  + \longip{\xi_1}{\H_0 + \Q_4}{\xi_1} + \mcE_{\Q_2}(\psi) + \mcE_{\textnormal{scat}} (\psi),
\end{aligned}
\end{equation*}
where 
\begin{align}
\mcE_{\Q_2}(\psi) 
& 
= 
    \int_0^1 \ud \lambda\, \longip{\xi_\lambda}{\Q_{2;B}^\mcE}{\xi_\lambda}
  - \int_0^1 \ud \lambda\int_\lambda^1 \ud \lambda'  \longip{\xi_{\lambda'}}{\Q_{2;B}^\mcE}{\xi_{\lambda'}},
 \label{eqn.def.mcE_Q2}
\\ 
\mcE_{\textnormal{scat}} (\psi)
& 
= \int_0^1 \ud \lambda\, \longip{\xi_\lambda}{[\H_0 + \Q_4,B] + \Q_2}{\xi_\lambda}.
 \label{eqn.def.mcE_scat}
\end{align}
We will derive bounds on these terms in Propositions~\ref{prop.mcE_Q2} and~\ref{prop.mcE_scat} below. 
In conclusion then (for any $N$-particle state $\psi$)
\begin{equation}\label{eqn.main}
\expect{H_N}_\psi 
  = E_F + \expect{\!\ud\Gamma(V(1-\varphi))}_F + \longip{\xi_1}{\H_0 + \Q_4}{\xi_1} + \mcE_V(\psi) + \mcE_{\Q_2}(\psi) + \mcE_{\textnormal{scat}}(\psi).
\end{equation}

With this formula we give the 

\begin{proof}[{Proof of \Cref{thm.main}}]
First, we calculate that 
$E_F = \frac{3}{5} N k_F^2 (1 + O(N^{-1/3}))$ 
by viewing the sum $E_F = \sum_{k\in B_F} |k|^2$ as a Riemann-sum for the corresponding integral $\frac{L^3}{(2\pi)^3} \int_{|k|\leq k_F} |k|^2 \ud k$.
Next, we calculate 
$\expect{\!\ud\Gamma(V(1-\varphi))}_F$.
With the aid of \Cref{a3} we find 
(since $V$ has compact support) that 
\begin{equation*}
\int V(1-\varphi) |x|^2 \ud x 
  = \int_{|x|\leq k_F^{-1}} V (1-\varphi_0) |x|^2 \ud x = 24 \pi a^3
\end{equation*}
for small $k_F$. 
Moreover, the two-particle reduced density of the free Fermi state satisfies \cite[Lemma 2.14]{Lauritsen.Seiringer.2024}
\begin{equation}\label{rho2b}
\rho^{(2)}(x,y) = \frac{1}{5(6\pi^2)^2} k_F^8 |x-y|^2 \left(1 + O(k_F^2|x-y|^2) + O(N^{-1/3})\right)
\end{equation}
(with $|x-y|$ denoting the metric on the torus.)
Thus, 
\begin{equation}\label{eqn.calc.<Vphi>.F}
\begin{aligned}
\expect{\!\ud\Gamma(V(1-\varphi))}_F
  & = \frac{1}{2} \iint V(x-y)(1-\varphi(x-y)) \rho^{(2)}(x,y) \ud x \ud y 
  \\ & = \frac{2}{5\pi} N a^3 k_F^5(1 + O( (ak_F)^2) + O(N^{-1/3})).
\end{aligned}
\end{equation}

For the third term in \Cref{eqn.main}, we can simply use  $\H_0 + \Q_4 \geq 0$.
Bounding the error terms $\mcE_V(\psi), \mcE_{\Q_2}(\psi)$ and $\mcE_{\textnormal{scat}}(\psi)$ occupies the main part of this paper.
We shall show that they are bounded as follows.

\begin{prop}
\label{prop.mcE_V}
\label{eqn.mcE_V.main}
Let $\psi$ be an approximate ground state. Then 
\begin{equation*}
\abs{\mcE_{V}(\psi)}
  \leq C N a^3 k_F^5 (ak_F)^{3/4}.
\end{equation*}
\end{prop}

\begin{prop}
\label{prop.mcE_Q2}
\label{eqn.mcE_Q2.main}
Let $\psi$ be an approximate ground state. Then 
\begin{equation*}
\abs{\mcE_{\Q_2}(\psi)}
  \leq C N a^3 k_F^5 (ak_F)^{3/2}.
\end{equation*} 
\end{prop}

\begin{prop}
\label{prop.mcE_scat}
\label{eqn.mcE_scat.main}
Let $\psi$ be an approximate ground state. Then 
\begin{equation*}
\abs{\mcE_{\textnormal{scat}}(\psi)}
  \leq C N a^3 k_F^5 (ak_F)^{3/10} \abs{\log ak_F}.
\end{equation*}
\end{prop}

Noting that the ground state is in particular an approximate ground state, 
we conclude the proof of \Cref{thm.main}.
\end{proof}

The remainder of the paper deals with the proofs of \Cref{prop.mcE_V,prop.mcE_Q2,prop.mcE_scat}.

\subsubsection*{Structure of the paper}
First, in \Cref{sec.a.priori}, we prove some useful a priori bounds and give the proof of \Cref{eqn.mcE_V.main}.
Next, in \Cref{sec.calc.commutators}, we calculate the commutators $[\H_0,B]$, $[\Q_4,B]$ and $[\Q_2,B]$ 
and extract the claimed leading terms. 
Then, in \Cref{sec.bdd.error-terms}, we bound the error terms arising from the commutators in a general state $\xi$.
Finally, in \Cref{sec.proof.final}, 
we use the bounds of \Cref{sec.bdd.error-terms} for the particular states $\xi_\lambda$ 
to prove \Cref{eqn.mcE_scat.main,eqn.mcE_Q2.main}.
Additionally we shall give the proof of  \Cref{prop.upper.bdd}.

In \Cref{sec.d=2} we sketch how to adapt the proof to the two-dimensional setting.

\section{A Priori Bounds}\label{sec.a.priori}
We begin our analysis with some a priori bounds on the operators $\mcN, \H_0, \Q_2$ and $\Q_4$  
(defined in \Cref{eqn.KE.particle-hole,eqn.def.Q2,eqn.def.Q4} above)
and  on the scattering function $\varphi$ in \Cref{eqn.def.scat.fun.phi}.

\subsection{A priori analysis of the kinetic energy}
We may bound $\expect{\mcN}_{R^*\psi}$ and $\expect{\H_0}_{R^*\psi}$ for an approximate ground state $\psi$ 
exactly as in \cite[Lemma 3.5 and Corollary 3.7]{Falconi.Giacomelli.ea.2021} using the positivity of the interaction.
We give the arguments here for completeness.

\begin{lemma}[{A priori bound on $\H_0$}]\label{lem.a.priori.H0}
Let $\psi$ be an approximate ground state. Then 
\begin{equation*}
\expect{\H_0}_{R^*\psi} \leq C N a^3 k_F^5.
\end{equation*}
\end{lemma}
\begin{proof}
By the positivity of the interaction we have (recalling \Cref{eqn.KE.particle-hole})
\begin{equation*}
\begin{aligned}
\expect{H_N}_\psi
=
\expect{\mcH}_{\psi} 
  & \geq \expect{\!\ud\Gamma(-\Delta)}_{\psi}
  = \expect{R^*\!\ud\Gamma(-\Delta)R}_{R^*\psi} 
  = E_F + \expect{\H_0}_{R^*\psi}.
\end{aligned}   
\end{equation*}
Then,
\begin{equation*}
\expect{\H_0}_{R^*\psi} \leq \expect{H_N}_\psi - E_F \leq C N a^3 k_F^5
\end{equation*}
by assumption of $\psi$ being an approximate ground state.
\end{proof}

\begin{lemma}[{A priori bound on $\mcN$}]\label{lem.a.priori.mcN}
Let $\psi$ be an approximate ground state. Then 
\begin{equation*}
\expect{\mcN}_{R^*\psi} \leq C N (ak_F)^{3/2}.
\end{equation*}
\end{lemma}
\begin{proof}
For any $\lambda > 0$ we can bound
\begin{equation*}
\begin{aligned}
\mcN = \sum_{k} a_k^* a_k 
  & \leq \sum_{k: ||k|^2 - k_F^2| \leq \lambda} 1 
 %   \\ &\quad 
 + \lambda^{-1} \sum_{k: ||k|^2 - k_F^2| > \lambda}  ||k|^2 - k_F^2| a_k^{*} a_k
  \\
  & \leq C N \lambda k_F^{-2} + \lambda^{-1} \H_0.
\end{aligned}
\end{equation*}
Using the bound on $\H_0$ from \Cref{lem.a.priori.H0} and choosing the optimal $\lambda = k_F^2 (ak_F)^{3/2}$ we conclude the desired bound.
\end{proof}

\subsection{A priori analysis of the interaction}\label{sec.a.priori.interaction}
Next, we study the interaction $V$ and prove a priori bounds on the operators $\Q_2$ and $\Q_4$ (defined in \Cref{eqn.def.Q2,eqn.def.Q4})
and bound the error-term $\mcE_V(\psi)$ from \Cref{eqn.main.expand.V=Q2+Q4}.

\begin{lemma}\label{lem.ini.V=Q2+Q4}
Let $\psi\in \mcF$ be any state. Then 
\begin{equation*}
\expect{\!\ud\Gamma(V)}_{\psi} = \expect{\!\ud\Gamma(V)}_F + \expect{\Q_2 + \Q_4}_{R^*\psi} + \mcE_V(\psi),
\end{equation*}
with
\begin{equation*}
\begin{aligned}
\abs{\mcE_V(\psi)} 
  & \lesssim 
    \eps N a^3 k_F^5 
    + \eps \expect{\Q_4}_{R^*\psi}
    + \eps^{-1} a^3 k_F^5 \expect{\mcN}_{R^*\psi} 
    + \eps^{-1} a^3 k_F^3 \expect{\H_0}_{R^*\psi}
\end{aligned}
\end{equation*}
for any $0 < \eps < 1$.
\end{lemma}

\begin{proof}
Recall \Cref{def.u.v.proj} and the short-hand notation $uu$ for $u \otimes u$. Viewing $V$ as a two-particle operator, we can write
\begin{equation*}
V = [vv + vu + uv + uu] V [vv + vu + uv + uu] 
\end{equation*}
and expand the product. 
Since $V\geq 0$ we may bound the cross-terms as 
\begin{equation*}
\begin{aligned}
\pm (uu V (vu+uv) + \hc) & \leq \eps uu V uu + \frac{1}{\eps} (vu + uv)V (vu+uv)
\\
\pm (vv V (vu+uv) + \hc) & \leq \eps vv V vv + \frac{1}{\eps} (vu+uv)V (vu+uv)
\end{aligned}
\end{equation*}
for any $\eps > 0$. Then we have the bound 
\begin{equation}\label{bound:v}
V \geq (1-\eps) vvVvv + vvVuu + uuVvv + (1-\eps) uuVuu + \left(1 - \frac{2}{\eps}\right)(vu+uv)V(vu+uv).
\end{equation}
(An analogous upper bound on $V$ holds if the signs in front of all $\eps$-dependent terms are reversed.)
Next, we shall write \eqref{bound:v} in second quantization,  conjugate by the particle-hole  transformation $R$ in \eqref{def:ph}, 
 use \Cref{eqn.prop.particle-hole.R} and normal order the terms. 
We leave out the details of the computations; 
they can be found for instance in \cite[Proposition 4.1]{Hainzl.Porta.ea.2020}.
The result is 
\begin{equation}\label{eqn.bdd.R*VR}
R^* \!\ud \Gamma(V) R 
  \geq (1-\eps) \expect{\!\ud\Gamma(V)}_F + \Q_2 + (1-\eps)\Q_4 + \X_\eps + \Q_{\eps}.
\end{equation}
(Again, an analogous upper bound holds reversing the sign of $\eps$.) Here $\Q_2,\Q_4$ are as in \Cref{eqn.def.Q2,eqn.def.Q4} and 
\begin{align}
\X_\eps & = \frac{1-\eps}{L^3} \sum_{\substack{p,q \\ p,q \in B_F}} [\hat V(p-q) - \hat V(0)] c_{q}^* c_q
  - \frac{1 - 2/\eps}{L^3} \sum_{\substack{p,q \\ p \in B_F \\ q \notin B_F}} [\hat V(p-q) - \hat V(0)] b_{q}^* b_q  ,
  \nonumber
\\
\Q_{\eps}
  & = \frac{1-\eps}{2} \iint V(x-y) c_x^* c_y^* c_y c_x \ud x \ud y 
  \nonumber 
  \\ & \quad 
    +  \left(\frac{1}{2} - \frac{1}{\eps}\right) \iint V(x-y) 
    	[ b_x^* c_x^* c_y b_y + b_y^* c_y^* c_x b_x - b_x^* c_y^* c_y b_x - b_y^* c_x^* c_x b_y] \ud x \ud y.
    \label{eqn.def.Qeps}
\end{align}
In particular, for any $\eps>0$
\begin{equation*}
\abs{\mcE_V(\psi)} 
	\leq \eps \expect{\!\ud\Gamma(V)}_F + \eps \expect{\Q_4}_{R^*\psi} + \abs{\expect{\X_{\pm\eps}}_{R^*\psi}} + \abs{\expect{\Q_{\pm\eps}}_{R^*\psi}}
\end{equation*}
(where the $\pm \eps$ has to be interpreted as taking the maximum over the two terms, either with $\eps$ or $-\eps$). 
To bound this we first note that $\expect{\!\ud\Gamma(V)}_F \leq C N a^3 k_F^5$. 
Indeed, by a simple Taylor-expansion, the two-particle density of the free Fermi gas is bounded by 
$\rho^{(2)}(x,y) \leq C k_F^8 |x-y|^2$, see \cite[Lemma 2.14]{Lauritsen.Seiringer.2024} or \Cref{rho2b} above.
Hence 
\begin{equation*}
\expect{\!\ud\Gamma(V)}_F = \frac{1}{2} \iint V(x-y) \rho^{(2)}(x,y) \ud x \ud y 
	\leq C k_F^8 \iint V(x-y) |x-y|^2 \ud x \ud y \leq C N a^3 k_F^5.
\end{equation*}
(Here $\int |x|^2 V \ud x$ is a finite constant of dimension $(\textnormal{length})^3$ and hence we write it as $C a^3$, 
as discussed in \Cref{rmk.length=a.in.bdds}.)
We are left with bounding the operators $\X_{\pm\eps}$ and $\Q_{\pm\eps}$.

To bound $\X_{\pm\eps}$ view $\hat V(k) = \int_\Lambda V(x) e^{-ikx} \ud x = \int_{\R^3} V(x) e^{-ikx}$ (for $L$ sufficiently large, by the compact support of $V$)
as defined on all of $\R^3$ (i.e., not just on $\frac{2\pi}{L}\Z^3$) and Taylor expand.
We find $\hat V(k) = \hat V(0)+ O(|k|^2 \int_{\R^3} |x|^2V(x) \ud x) = \hat V(0) + O(a^3|k|^2)$
since $V$ is a real-valued radial function.
Thus, for any state $\xi\in \mcF$,
\begin{equation*}
\begin{aligned}
\abs{\expect{\X_{\pm\eps}}_\xi}
  & \lesssim \frac{a^3 }{L^3}  \sum_{k\in B_F} \sum_{\ell \in B_F} |k-\ell|^2 \expect{c_\ell^*c_\ell}_\xi 
  \\ & \quad 
    + \eps^{-1} \frac{a^3}{L^3} 
    \left[\sum_{k\in B_F} |k|^2 \sum_{\ell} \expect{b_\ell^*b_\ell}_\xi 
    + \sum_{k\in B_F} \sum_{\ell} |\ell|^2 \expect{b_\ell^*b_\ell}_\xi \right]
  \\ & 
    \lesssim
    \eps^{-1} a^3 k_F^5 \expect{\mcN}_\xi + \eps^{-1} a^3 k_F^3 \expect{\H_0}_\xi
\end{aligned}
\end{equation*}
(for any $0 < \eps < 1$) by noting that $\sum_k |k|^2 a_k^* a_k \leq \H_0 + k_F^2 \mcN$.

Finally, we shall show that $\Q_{\pm\eps}$ is suitably small. Note that the first term in \Cref{eqn.def.Qeps} is actually non-negative for $0\leq \eps\leq 1$ and could be dropped for a lower bound; since we will also derive an upper bound with our method, however, we shall not do this and estimate it instead. The main idea is to ``Taylor expand in $x-y$''. 
More precisely, for the first term in \Cref{eqn.def.Qeps} we write 
\begin{equation}
\label{eqn.Taylor.c_x.1st.order}
c_y = c_x + (y-x) \cdot \int_0^1 \nabla c_{x + t(y-x)} \ud t.
\end{equation}
The integration has to be understood as connecting the two points $x$ and $y$ by the shortest line on the torus. 
We apply this to the factors $c_y^*$ and $c_y$.
Noting that $c_x^2 = 0$ (recall that $c_x$ is in fact a bounded operator) 
and changing variables to $z = y-x$  we then find (recalling that $V$ is even)  
\begin{equation*}
\iint V(x-y) c_x^* c_y^* c_y c_x \ud x \ud y 
	= \iint V(z) z^\mu z^\nu 
      \int_0^1 \ud t\int_0^1 \ud s \, 
      c_x^* \nabla^\mu c_{x+tz}^* \nabla^\nu c_{x+sz} c_x 
      \ud x \ud z.
\end{equation*}
In norm we have $\norm{\nabla c_x}\leq C k_F^{3/2+1}$.
Thus, the expectations value of the first term in \Cref{eqn.def.Qeps} in any state $\xi$ is bounded by 
\begin{equation*}
\abs{\frac{1\mp \eps}{2} \iint V(x-y) \expect{c_x^* c_y^* c_y c_x}_\xi \ud x \ud y} 
	\leq  C k_F^5 \int |z|^2 V(z) \ud z \int \norm{c_x \xi}^2 \ud x.
\end{equation*}
Furthermore, 
\begin{equation*}
\int \norm{c_x\xi}^2 \ud x 
  = \sum_{k\in B_F} \longip{\xi}{a_k^*a_k}{\xi} 
  \leq \expect{\mcN}_\xi.
\end{equation*}

Denote the second term of \Cref{eqn.def.Qeps} by $\Q_{\eps,2}$ and define for any state $\xi$ the function
$\phi(x,y) = \expect{b_x^* c_x^* c_y b_y + b_y^* c_y^* c_x b_x - b_x^* c_y^* c_y b_x - b_y^* c_x^* c_x b_y}_\xi$.
We note that $\phi(y,x) = \phi(x,y)$ and $\phi(x,x) = 0$. Thus, by Taylor-expansion in $y$ around $y=x$ 
the zeroth and first orders vanish, and hence  
\begin{equation*}
\phi(x,y) = (y-x)^\mu (y-x)^\nu \int_0^1 \ud t \, (1-t) [\nabla_2^\mu \nabla_2^\nu \phi](x,x+t(y-x)).
\end{equation*}
(Here $\nabla^\mu_2$ denotes a derivative in the second variable.)
Thus (changing variables to $z=x-y$ and using that $V$ is even) we have 
\begin{equation*}
\expect{\Q_{\pm\eps,2}}_\xi
	= \left(\frac{1}{2} \mp \frac{1}{\eps}\right) 
		\int_0^1 \ud t \, (1-t) 
		\int \ud z \, z^\mu z^\nu V(z) 
		\int \ud x \, [\nabla_2^\mu \nabla_2^\nu \phi](x,x+tz).
\end{equation*}
The $x$-integral is given by (with the derivatives now being with respect to $x$)
\begin{equation*}
\begin{aligned}
& \int \ud x \,
	\Bigl< b_x^* c_x^* [ \nabla^\mu \nabla^\nu c_{x+tz} b_{x+tz}] 
		+ [\nabla^\mu \nabla^\nu  b_{x+tz}^* c_{x+tz}^*] c_x b_x
		- b_x^*[ \nabla^\mu \nabla^\nu c_{x+tz}^* c_{x+tz}]  b_x
	\\ & \qquad 
		- [\nabla^\mu \nabla^\nu b_{x+tz}^*] c_x^* c_x b_{x+tz}
		- b_{x+tz}^* c_x^* c_x [\nabla^\mu\nabla^\nu b_{x+tz}]
	% \\ & \qquad 
		- ( [\nabla^\mu b_{x+tz}^*] c_x^* c_x [\nabla^\nu b_{x+tz}] + (\mu \leftrightarrow \nu))
	\Bigr>_\xi
\\ & 
% \quad 
	= \int \ud x \, 
	\Bigl<
		- [\nabla^\mu b_x^* c_x^*] [\nabla^\nu c_{x+tz} b_{x+tz}]  
		- [\nabla^\mu b_{x+tz}^* c_{x+tz}^*] [\nabla^\nu c_{x} b_{x}] 
		- b_x^* [\nabla^\mu \nabla^\nu c_{x+tz}^* c_{x+tz}]  b_x
	\\ & \qquad 
		+ [\nabla^\mu b_{x+tz}^*]  [\nabla^\nu c_x^* c_x b_{x+tz}]
		+ [\nabla^\mu b_{x+tz}^* c_x^* c_x] [\nabla^\nu b_{x+tz}]
	% \\ & \qquad 
		- ([\nabla^\mu b_{x+tz}^*] c_x^* c_x [\nabla^\nu b_{x+tz}] + (\mu \leftrightarrow \nu))
	\Bigr>_\xi
\end{aligned}
\end{equation*}
by integration by parts. Using the norm  estimate $\norm{\nabla^n c_x}\leq C k_F^{3/2+n}$ with $\nabla^n$ denoting any $n$'th derivative we thus bound the $x$-integral with the aid of Cauchy--Schwarz as
\begin{equation*}
\begin{aligned}
\lesssim 
 k_F^3 \int \ud x \norm{\nabla b_x \xi}^2 
 + k_F^5 \int \ud x \norm{b_x\xi}^2.
\end{aligned}
\end{equation*}
Noting that 
\begin{equation*}
\begin{aligned}
\int \norm{b_x\xi}^2 \ud x 
  & 
  = \sum_{k\notin B_F} \longip{\xi}{a_k^*a_k}{\xi} 
  \leq \expect{\mcN}_\xi,
% \qquad 
\\
  \int \norm{\nabla b_x \xi}^2 \ud x 
    & 
    = \sum_{k\notin B_F} |k|^2 \longip{\xi}{a_k^* a_k}{\xi} 
    % \\ & 
    \leq \expect{\H_0}_\xi + k_F^2 \expect{\mcN}_\xi
\end{aligned}
\end{equation*}
we thus obtain 
\begin{equation*}
\abs{\expect{\Q_{\pm \eps,2}}_\xi} 
	\lesssim  \eps^{-1} a^3 k_F^5 \expect{\mcN}_\xi + \eps^{-1} a^3 k_F^3 \expect{\H_0}_\xi.
\end{equation*}
Together with the bounds above, this concludes the proof of the lemma.
\end{proof}

Additionally, we need a priori bounds on $\Q_2$ and $\Q_4$ (defined in \Cref{eqn.def.Q2,eqn.def.Q4}).
These follow from arguments similar to the analogous bounds in \cite[Proposition 3.3(e) and Lemma 3.9]{Falconi.Giacomelli.ea.2021}.

\begin{lemma}[{A priori bound for $\Q_2$}]\label{lem.a.priori.Q2}
For suitable $C>0$ we have 
\begin{equation*}
\abs{\expect{\Q_2}_{\xi}} \leq \delta \expect{\Q_4}_{\xi} + C \delta^{-1} N a^3 k_F^5.
\end{equation*}
for any $\delta > 0$ and any state $\xi\in \mcF$.
\end{lemma}
\begin{proof}
By Cauchy--Schwarz we have for any $\delta > 0$
\begin{equation*}
\begin{aligned}
2\abs{\expect{\Q_2}_\xi}
  & \leq \iint V(x-y) \abs{\expect{b_x^* b_y^* c_y^* c_x^* + \hc}_\xi} \ud x \ud y
  \\
  & \leq 
    \iint V(x-y) \left(\delta \norm{b_y b_x \xi}^2 + \delta^{-1} \norm{c_y^* c_x^* \xi}^2\right) \ud x \ud y
  \\
  & \leq 2\delta \expect{\Q_4}_\xi + \delta^{-1} \iint V(x-y) \norm{c_y^* c_x^*}^2 \ud x \ud y.
\end{aligned}
\end{equation*}
We Taylor expand $c_y^*$ as in \Cref{eqn.Taylor.c_x.1st.order} above. 
Being fermionic operators $(c_x^*)^2 = 0$. 
In particular we have $\norm{c_y^* c_x^*}\leq C k_F^4 |x-y|$ since $\norm{c_x^*}\leq C k_F^{3/2}$ and $\norm{\nabla c_x^*}\leq C k_F^{3/2+1}$. 
With $\int |x|^2 V \ud x \leq C a^3$ we conclude the proof of the lemma. 
\end{proof}

\begin{lemma}[{A priori bound for $\Q_4$}]\label{lem.a.priori.Q4}
Let $\psi$ be an approximate ground state. Then 
\begin{equation*}
\expect{\Q_4}_{R^*\psi} \leq C N a^3 k_F^5.
\end{equation*}
\end{lemma}

\begin{proof}
Recalling \Cref{eqn.KE.particle-hole}, applying \Cref{lem.ini.V=Q2+Q4} and \Cref{lem.a.priori.Q2} with $\delta = 1/2$, 
and using $\expect{\!\ud\Gamma(V)}_F \geq 0$ we find 
\begin{equation*}
\begin{aligned}
\expect{H_N}_\psi 
  & \geq E_F 
  + \expect{\H_0}_{R^*\psi} 
  + \left(\frac{1}{2} - \eps\right) \expect{\Q_4}_{R^*\psi}
  - C N a^3 k_F^5 
  \\ & \quad 
  - C \eps^{-1} a^3 k_F^5 \expect{\mcN}_{R^*\psi}
  - C \eps^{-1} a^3 k_F^3 \expect{\H_0}_{R^*\psi}.
\end{aligned}
\end{equation*}
Using the a priori bounds of \Cref{lem.a.priori.H0,lem.a.priori.mcN} 
and noting that $\H_0 \geq 0$ we obtain the desired result (by taking $\eps=1/4$, say).
\end{proof}

With the results of this section we can give the

\begin{proof}[Proof of \Cref{eqn.mcE_V.main}]
Combining \Cref{lem.ini.V=Q2+Q4,lem.a.priori.H0,lem.a.priori.Q4,lem.a.priori.mcN}
and choosing the optimal $\eps = (ak_F)^{3/4}$ we conclude the proof of \Cref{prop.mcE_V}.
\end{proof}

\subsection{The scattering function}\label{sec.a.priori.scat.fun}
We now prove some a priori bounds on the scattering function $\varphi$ defined in \Cref{eqn.def.scat.fun.phi} above.

\begin{lemma}[{Properties of $\varphi$}]\label{lem.prop.phi}
The scattering function $\varphi$ satisfies 
\begin{equation*}
\begin{aligned}
\norm{|\cdot|^n \varphi}_{L^1}
& \leq 
C a^3 k_F^{-n}, 
& n &= 1,2,
\quad  
&
\norm{|\cdot|^n \nabla^n \varphi}_{L^1}
& \leq C a^3 \abs{\log ak_F}, 
&n&=0,1,2
\quad  
\\
\norm{ |\cdot| \varphi}_{L^2 }
& \leq 
C a^{3/2+1}, 
&&
\quad 
&
\norm{|\cdot|^n \nabla^n \varphi}_{L^2}
& \leq C a^{3/2}, 
&n&=0,1.
\end{aligned}
\end{equation*}
Here $\nabla^n$ represents any combination of $n$ derivatives
and $\abs{\cdot}$ denotes the metric on the torus in the sense that $\abs{x}$ is the distance between $x$ and the point $0$.

\end{lemma}

\begin{remark}\label{eqn.bdd.mcE_phi.L1}
Recalling the bound in \Cref{eqn.bdd.mcE_varphi} for $\mcE_\varphi$ we see that $\mcE_\varphi$ 
satisfies the same bounds as $\varphi$ only with an additional power $k_F$.  
\end{remark}

\begin{proof}
Recall that the scattering function is given by (the periodization of) $\varphi = \varphi_0 \chi_\varphi(k_F|\cdot|)$. Since $\varphi_0$ is a radial function, so is $\varphi$ (rather, $\varphi_0 \chi_\varphi(k_F|\cdot|)$ is).
In particular, for the bounds we may replace $\nabla^n$ by $\partial_r^n$, 
with $\partial_r$ the derivative in the radial direction.

We first establish the bound 
\begin{equation}\label{eqn.bdd.nabla.phi}
  \abs{\partial_r \varphi_0(x)} \leq \frac{Ca^3}{|x|(a^3+|x|^3)},
  \qquad \textnormal{for all } x\in \R^3.
\end{equation}
To prove this, we note that, in radial coordinates, the scattering equation \eqref{eqn.scat} reads
\begin{equation*}
  r \partial_r^2 \varphi_0 + 4 \partial_r \varphi_0 + \frac{1}{2} r V (1-\varphi_0) = 0.
\end{equation*}
Further, we recall that 
$0 \leq \varphi_0(x) \leq \min\{1, a^3 / |x|^3\}$ for all $x$
and that $\varphi_0$ is a radially decreasing function, which follows from \cite[Lemma A.1]{Lieb.Yngvason.2001} applied to $5$ dimensions. 
Using the scattering equation we then compute 
\begin{equation*}
  \partial_r \left[\partial_r \varphi_0 + \frac{4}{r} \varphi_0 \right]
    = \partial_r^2 \varphi_0 + \frac{4}{r} \partial_r \varphi_0 - \frac{4}{r^2} \varphi_0
    = - \frac{4}{r^2} \varphi_0 - \frac{1}{2} V (1-\varphi_0) \leq 0.
\end{equation*}
Integrating and recalling that $\varphi_0(x) = a^3/|x|^3$ outside the support of $V$ we find 
\begin{equation*}
  \partial_r \varphi_0 + \frac{4}{r} \varphi_0 \geq 0.
\end{equation*}
Since $\partial_r\varphi_0 \leq 0$ this proves \Cref{eqn.bdd.nabla.phi}.
From \Cref{eqn.bdd.nabla.phi} and the scattering equation \eqref{eqn.scat} we immediately find 
\begin{equation*}
\abs{\partial_r^2\varphi_0} 
  \lesssim \frac{a^3}{|x|^2 (a^3 + |x|^3)} + V.
\end{equation*}

To extract the bounds on $\varphi$ 
we note that $\norm{\partial^n \chi_\varphi}_{L^\infty} \leq C$ for any $n$, 
since $\chi_\varphi$ is smooth by assumption. 
Furthermore, also by assumption, $V\in L^1$. Finally, $\varphi$ is supported only for $|x|\leq 2k_F^{-1}$ 
since $\chi_\varphi(t) = 0$ for $t> 2$. 
The bounds immediately follow.
\end{proof}

\section{Calculation of Commutators}\label{sec.calc.commutators}
In this section we shall calculate the commutators $[\H_0,B]$, $[\Q_4,B]$ and $[\Q_2,B]$ 
and find formulas for the error-terms $\mcE_{\Q_2}(\psi)$ and $\mcE_{\textnormal{scat}}(\psi)$ in \Cref{eqn.def.mcE_Q2,eqn.def.mcE_scat}.
The computations of the commutators are analogous to those in \cite{Falconi.Giacomelli.ea.2021,Giacomelli.2023},
where similar commutators are computed.

\subsection{\texorpdfstring{$[\H_0,B]$}{[H0,B]}:}
We have (recalling the formula for $\H_0$ in \Cref{eqn.KE.particle-hole} and using the representation of $B$ in momentum-space in \Cref{eqn.def.B})
\begin{equation*}
[ \H_0, B]	
	= \frac{1}{2L^3} \sum_{p,q,k,k'} ||q|^2 - k_F^2| \hat\varphi(p) [a_q^* a_q, b_{k+p}^r b_{k'-p}^r c_{k'} c_k] + \hc.
\end{equation*}
To compute the commutator 
we note that $[a_q^* a_q, a_k] = - \delta_{q,k} a_k $, hence
\begin{equation*}
\begin{aligned}
[a_q^* a_q, b_{k+p}^r b_{k'-p}^r c_{k'} c_k] 
	& = 
	- (\delta_{q,k+p} + \delta_{q,k'-p} + \delta_{q,k'} + \delta_{q,k}) b_{k+p}^r b_{k'-p}^r c_{k'} c_k.
\end{aligned}
\end{equation*}
Next, for $k+p,k'-p\notin B_F$ and $k,k'\in B_F$ we have 
\begin{equation*}
\begin{aligned}
||k+p|^2 - k_F^2| + ||k'-p|^2 - k_F^2| + ||k|^2 - k_F^2| + ||k'|^2 - k_F^2|
&
= |k+p|^2 + |k'-p|^2 -|k|^2 -|k'|^2 
\\ & = 2 |p|^2 + 2 p\cdot (k-k').
\end{aligned}
\end{equation*}
We conclude that 
\begin{equation*}
\begin{aligned}
  [\H_0, B] & = 
    - \frac{1}{L^3} \sum_{p,k, k'} \left(|p|^2 + p \cdot (k - k')\right) \hat\varphi(p) 
    b_{k+p}^r b_{k'-p}^r c_{k'} c_k + \hc.
\end{aligned}     
\end{equation*}
By the symmetry $k\leftrightarrow k'$ and $p\to -p$ we may write the second summand as 
\begin{equation*}
- \frac{2}{L^3} \sum_{p,k, k'} p\cdot k \hat\varphi(p)     
b_{k+p}^r b_{k'-p}^r c_{k'} c_k + \hc.
\end{equation*}
Writing then in configuration space we get (recall the Einstein summation convention, introduced in Notation~\ref{Enot}.)
\begin{equation*}
[\H_0, B]
  = \iint \left(\Delta \varphi(x-y) b_x^r b_y^r c_y c_x  + 2 \nabla^\mu \varphi(x-y) b_x^r b_y^r c_y \nabla^\mu c_x\right) \ud x \ud y + \hc.
\end{equation*}
We first replace the $b$'s by their non-regularized counterparts. That is, we write 
\begin{equation}\label{eqn.[H0.B].ini}
[\H_0, B]
  = \iint \left(\Delta \varphi(x-y) b_x b_y c_y c_x  + 2 \nabla^\mu \varphi(x-y) b_x b_y c_y \nabla^\mu c_x\right) \ud x \ud y + \hc
    + \H_{0;B}^{\div r}
\end{equation}
with $\H_{0;B}^{\div r}$ defined so that this holds. 
Write now (similarly as in \Cref{eqn.Taylor.c_x.1st.order})
\begin{align}
c_x 
  & = c_y + (x-y)^\mu \int_0^1 \ud t \, \nabla^\mu c_{y+t(x-y)} 
  \nonumber
  \\
  & = c_y + (x-y)^\mu \nabla^\mu  c_x - (x-y)^\mu (x-y)^\nu  \int_0^1 \ud t \, (1-t) \nabla^\mu \nabla^\nu c_{x + t(y-x)}.
\label{eqn.expand.cx.2nd.order}
\end{align}
Note that the first order term $\nabla^\mu c_x$ is evaluated at $x$ and not at $y$. 
We apply \Cref{eqn.expand.cx.2nd.order} to the factor $c_x$ in the first term in \Cref{eqn.[H0.B].ini} above. 
Being a fermionic operator $c_y^2 = 0$. 
Thus, we have 
\begin{equation}\label{eqn.calc.[H0.B]}
\begin{aligned}
[\H_0,B]
  & = \iint \left[(x-y)^\mu \Delta \varphi(x-y)   + 2 \nabla^\mu \varphi(x-y) \right] b_x b_y c_y \nabla^\mu c_x \ud x \ud y + \hc
    + \H_{0;B}^{\textnormal{Taylor}} + \H_{0;B}^{\div r},
\end{aligned}
\end{equation}
with $\H_{0;B}^{\textnormal{Taylor}}$ as in the first term of \Cref{eqn.[H0.B].ini} only with the factor $c_x$ replaced by 
the last term of \Cref{eqn.expand.cx.2nd.order}. 
It is an error term.

\subsection{\texorpdfstring{$[\Q_4,B]$}{[Q4,B]}:}
Recall the definition of $\Q_4$ from \Cref{eqn.def.Q4} and of $B$ from \cref{eqn.def.B}. 
We compute 
\begin{equation*}
[\Q_4,B]
  = \frac{1}{4} \iiiint \ud x \ud y \ud z \ud z' \, V(x-y) \varphi(z-z') 
   [ b_x^* b_y^* b_y b_x, b_z^r b_{z'}^r c_{z'} c_z]
   + \textnormal{h.c.}
\end{equation*}
The commutator is (using that $b$'s and $c$'s have disjoint support in momentum-space so they anti-commute: $\{b_x^*, c_y\}=0$)
\begin{equation*}
[ b_x^* b_y^* b_y b_x, b_z^r b_{z'}^r c_{z'} c_z  ]
  = [ b_x^* b_y^*, b_z^r b_{z'}^r  ] b_y b_x c_{z'} c_z.
\end{equation*}
Using that 
\begin{equation*}
 \{b_x^*, b_y^r\} = \frac{1}{L^3} \sum_{k} \hat u^r (k) e^{ik(y-x)} = u^r(y-x) = u^r(x-y)
\end{equation*}
one computes 
\begin{equation}\label{eqn.calc.[bb.bb]}
\begin{aligned}
[b_{x}^* b_{y}^*, b_{z}^r b_{z'}^r ]
& = 
   u^{r}(x-z) u^{r}(y-z') - u^{r}(y-z)u^{r}(x-z') 
\\ & \quad 
  + u^{r}(y-z) b_x^* b_{z'}^r + u^{r}(x-z') b_y^* b_{z}^r
  - u^{r}(x-z) b_y^* b_{z'}^r - u^{r}(y-z') b_x^* b_z^r
\end{aligned}
\end{equation}
The two first terms, respectively the four last terms, give the same contributions to $[\Q_4,B]$ when the integration is performed by the symmetries 
$x\leftrightarrow y$ and $z\leftrightarrow z'$.
We write the $u^r$'s as $u^r = \delta - v^r$.
The quartic term with two $\delta$'s is the main term. 
In conclusion then 
\begin{equation}\label{commQ4}
\begin{aligned}
[\Q_4, B]
  & = 
    - \frac{1}{2} \iint V(x-y) \varphi(x-y) b_x b_y c_y c_x \ud x \ud y 
  + \hc
  + \Q_{4;B}^{\mcE}
\end{aligned}
\end{equation}
with 
\begin{equation}\label{eqn.def.[Q4.B]-error}
\begin{aligned}
  \Q_{4;B}^{\mcE}
  & =
    \frac{1}{2}
    \iiiint 
    V(x-y) \varphi(z-z') 
      \Bigl[
      \bigl(
      2\delta(y-z)v^r(x-z') 
      - v^r(x-z')v^r(y-z)
      \bigr)
      b_y b_x c_{z'} c_z
  \\ & \qquad  
      -
      2 (\delta(x-z) - v^r(x-z))b_y^* b_{z'}^r b_y b_x c_{z'} c_z
      \Bigr]
    \ud x \ud y \ud z \ud z' 
  + \hc.
\end{aligned}
\end{equation}
In the first term we rewrite $c_x$ using \Cref{eqn.expand.cx.2nd.order} as for the case of $\H_0$. 
Then we find
\begin{equation}\label{eqn.calc.[Q4.B]}
\begin{aligned}
[\Q_4, B]
  & = -\frac{1}{2}\iint (x-y)^\mu V(x-y) \varphi(x-y) b_x b_y c_y \nabla^\mu c_x \ud x \ud y + \hc + \Q_{4;B}^{\textnormal{Taylor}} + \Q_{4;B}^{\mcE}
\end{aligned}
\end{equation}
and $ \Q_{4;B}^{\textnormal{Taylor}}$ as in the first term of \eqref{commQ4} only with $c_x$ replaced by the last term of \Cref{eqn.expand.cx.2nd.order}.

\subsection{\texorpdfstring{$[\Q_2,B]$}{[Q2,B]}:}
Recall the definition of $\Q_2$ from \Cref{eqn.def.Q2} and of $B$ from \Cref{eqn.def.B}. 
We have 
\begin{equation*}
[\Q_2,B]
  = \frac{1}{4} \iiiint V(x-y) \varphi(z-z') [b_x^* b_y^* c_y^* c_x^*, b_z^r b_{z'}^r c_{z'} c_z] \ud x \ud y \ud z \ud z' + \hc.
\end{equation*}
The commutator is given by (recall again that $b$'s and $c$'s anti-commute)
\begin{equation}
\begin{aligned}
[b_{x}^* b_{y}^* c_{y}^* c_x^* , b_{z}^r b_{z'}^r c_{z'} c_z]
  &  =
    - [b_{x}^* b_{y}^*  , b_{z}^r b_{z'}^r ] [c_{y}^* c_x^*, c_{z'} c_z]
  % \\ & \quad 
    + b_{x}^* b_{y}^* b_{z}^r b_{z'}^r [c_{y}^* c_x^* ,  c_{z'} c_z]
  \\ & \quad
    +  c_{y}^* c_x^* c_{z'} c_z [b_{x}^* b_{y}^*  , b_{z}^r b_{z'}^r].
\end{aligned}
\label{eqn.com.bbcc.b*b*c*c*}
\end{equation}
The $b$-commutator is as in \Cref{eqn.calc.[bb.bb]}. Similarly
\begin{equation}\label{eqn.calc.[cc.cc]}
\begin{aligned}
[c_{y}^* c_{x}^*, c_{z'} c_{z} ]
& = 
  v(x-z) v(y-z') - v(y-z)v(x-z') 
\\ & \quad 
  + v(y-z) c_x^* c_{z'} + v(x-z') c_y^* c_{z}
  - v(x-z) c_y^* c_{z'} - v(y-z') c_x^* c_z
\end{aligned}
\end{equation}
The leading term is the constant (i.e. fully contracted) one. All other terms contribute to the error.
Furthermore, in the constant term we write $u^r = \delta - v^r$. Then the term with all $\delta$'s is the main one and the remainder are errors.
We conclude 
\begin{equation}\label{eqn.calc.[Q2.B]}
\begin{aligned}
[\Q_2,B]
  & = -\frac{1}{2} 
    \iiiint V(x-y) \varphi(z-z') 
    \left[\delta(x-z)\delta(y-z')  - \delta(x-z') \delta(y-z)\right]
  \\ & \hphantom{= -\frac{1}{2} \iiiint } \times 
    \left[v(x-z)v(y-z')  - v(x-z') v(y-z)\right]
    \ud x \ud y \ud z \ud z'
    + \Q_{2;B}^\mcE
  \\
  & 
    = - \iint V(x-y) \varphi(x-y) \rho^{(2)}(x,y) \ud x \ud y  + \Q_{2;B}^\mcE
    = - 2 \expect{\!\ud\Gamma(V\varphi)}_F + \Q_{2;B}^\mcE
\end{aligned}
\end{equation}
using that the two-particle density of the free Fermi state is given by $\rho^{(2)}(x,y) = |v(0)|^2 - |v(x-y)|^2$.

\section{Bounding Error Terms}\label{sec.bdd.error-terms}
To prove the bounds in \Cref{eqn.mcE_scat.main,eqn.mcE_Q2.main} we first  bound 
$[\H_0 + \Q_4, B] + \Q_2$ and $\Q_{2;B}^\mcE$.  
Then in \Cref{sec.proof.final} we use these bounds for the particular states $\xi_\lambda$ from \Cref{eqn.def.xi.lambda} 
with $\psi$ an approximate ground state 
to conclude the proofs of \Cref{eqn.mcE_scat.main,eqn.mcE_Q2.main}.

To bound $[\H_0 + \Q_4, B] + \Q_2$ we first recall the definition of $\Q_2$ in \Cref{eqn.def.Q2}
\begin{equation*}
\Q_2 = \frac{1}{2} \iint V(x-y) b_x b_y c_y c_x \ud x \ud y + \hc.
\end{equation*}
Rewriting the factor $c_x$ as in \Cref{eqn.expand.cx.2nd.order} we thus have 
\begin{equation*}
\Q_2 = \frac{1}{2} \iint (x-y)^\mu V(x-y) b_x b_y c_y \nabla^\mu c_x \ud x \ud y + \hc + \Q_2^{\textnormal{Taylor}}
\end{equation*}
with $\Q_2^{\textnormal{Taylor}}$ appropriately defined.
Then, recalling \Cref{eqn.calc.[H0.B],eqn.calc.[Q4.B]}, we have
\begin{equation}\label{eqn.[H0.B]+Q2.decompose}
\begin{aligned}
[\H_0 + \Q_4, B] + \Q_2 
  & = \Q_{\textnormal{scat}}
    + \H_{0;B}^{\textnormal{Taylor}} + \H_{0;B}^{\div r} + \Q_{4;B}^{\textnormal{Taylor}} + \Q_{4;B}^{\mcE} + \Q_2^{\textnormal{Taylor}}
\end{aligned}
\end{equation}
with (recalling \Cref{eqn.defn.mcE_varphi.scat.eqn})
\begin{equation}\label{eqn.def.Qscat}
\begin{aligned}
\Q_{\textnormal{scat}} 
% \\
  & = \iint \mcE_\varphi^\mu(x-y) b_x b_y c_y \nabla^\mu c_x \ud x \ud y + \hc.
\end{aligned}
\end{equation}
Below we bound separately each of the operators 
\begin{equation*}
% \label{eqn.list.error.operators}
\T := 
\H_{0;B}^{\textnormal{Taylor}} +  
\Q_{4;B}^{\textnormal{Taylor}} + 
\Q_2^{\textnormal{Taylor}},
\quad 
\Q_{\textnormal{scat}}, \quad 
\H_{0;B}^{\div r}, \quad 
\Q_{4;B}^\mcE
\quad 
\textnormal{and} 
\quad 
\Q_{2;B}^\mcE
\end{equation*}
with $\Q_{2;B}^\mcE$ from \Cref{eqn.calc.[Q2.B]}. 
We shall show that 
\begin{lemma}\label{lemma.bdd.list.operators}
For any state $\xi\in \mcF$, and any $\alpha > 0$, we have 
\begin{align*}
\abs{\expect{\T}_\xi}
+
\abs{\expect{\Q_{\textnormal{scat}}}_\xi}
& \lesssim 
N^{1/2} a^3 k_F^5 (ak_F)^{-3\alpha/2} \abs{\log ak_F}   \expect{\mcN}_\xi^{1/2} 
  + N^{1/2} a k_F^2 (ak_F)^{1/2+\alpha}  \expect{\H_0}_\xi^{1/2},
\\
\abs{\expect{\H_{0;B}^{\div r}}_\xi}
& \lesssim 
N^{1/2}  a^3 k_F^4 \abs{\log ak_F} \expect{\H_0}_{\xi}^{1/2} 
+ N^{1/2} a^3 k_F^5 \abs{\log ak_F} \expect{\mcN}_{\xi}^{1/2},
\\
\abs{\expect{\Q_{4;B}^\mcE}_\xi}
& \lesssim 
N^{1/2} a^3 k_F^4 (ak_F)^{1/2}  \expect{\Q_4}_\xi^{1/2}
  + a^3 k_F^4 \expect{\mcN}_\xi^{1/2}\expect{\Q_4}_\xi^{1/2},
\\
\abs{\expect{\Q_{2;B}^\mcE}_\xi}
& \lesssim 
a^2 k_F^2 \expect{\H_0}_\xi^{1/2} \expect{\Q_4}_{\xi}^{1/2}
+ a^3 k_F^5 \expect{\mcN}_\xi 
+ N a^5 k_F^7 \abs{\log ak_F}.
\end{align*}
\end{lemma}

The rest of this section deals with the proof of \Cref{lemma.bdd.list.operators}. 
We first give some preliminary bounds.

\subsection{Technical lemmas}
In the proof of \Cref{lemma.bdd.list.operators} below 
we shall need a bound on integrals of the form $\int F(y) b_y c_y \ud y$ and similar.
The following lemma will turn out to be very helpful.

\begin{lemma}\label{lem.bdd.F*ac.b(F)}
Let $F$ be a compactly supported function with $F(x)=0$ for $|x| \geq C k_F^{-1}$.
Then, uniformly in $x\in \Lambda$ and $t\in [0,1]$, (with $\nabla^n$ denoting any $n$'th derivative)
\begin{align}
\norm{\int F(x-y) a_y c_y \ud y} 
	& \lesssim k_F^{3/2} \norm{F}_{L^2} ,
	\label{eqn.bdd.Fac.b(F)}
	\\ 
\norm{\int F(x-y) a_y c_y \nabla^n c_{ty+(1-t)x} \ud y 	}
	& \lesssim k_F^{3+n} \norm{F}_{L^2}.
	\label{eqn.bdd.Facc.b(F)}
\end{align}
\end{lemma}

\begin{remark}
Recall that $\int F(y) a_y \ud y = a(\overline{F})$ has norm $\norm{a(\overline{F})} = \norm{F}_{L^2}$.
The lemma shows that for certain  bounded operators $A_y$ (being either $c_y$ or $c_y \nabla^n c_{ty+(1-t)x}$) $\norm{\int F(y) a_y A_y \ud y}$ can be bounded by (a constant times) $\norm{F}_{L^2} \sup_y \norm{A_y}$.
\end{remark}

\begin{remark}
A similar bound is given in \cite[Lemma 4.8]{Giacomelli.2023}. 
There, however, the right-hand side depends on $\norm{F}_{L^1}, \norm{\nabla F}_{L^1}$ and $\norm{\Delta F}_{L^1}$.
In \Cref{lem.bdd.F*ac.b(F)} we do not require any smoothness on $F$. 
This translates to having no smoothness assumptions on $V$ in \Cref{thm.main}.
\end{remark}

\begin{remark}
Applying \Cref{eqn.bdd.Fac.b(F)} for $F=\varphi$ we find (by writing $b^r_{z'} = a_{z'} - (a_{z'}- b^r_{z'})$ and 
noting that $\norm{a_{z'} - b_{z'}^r} \leq C k_F^{3/2}$ since both operators agree for  momenta $|k|\geq 3k_F$)
\begin{align}
\norm{\int \varphi(z-z') b^r_{z'} c_{z'} \ud z'} 
& \lesssim k_F^3 \norm{\varphi}_{L^1} + k_F^{3/2} \norm{\varphi}_{L^2}
\leq C (ak_F)^{3/2}
\label{eqn.bdd.phi*bc} 
\end{align}
where we have used \Cref{lem.prop.phi} to bound the norms of $\varphi$.
Moreover, by Taylor-expanding as in \Cref{eqn.Taylor.c_x.1st.order} we have
\begin{equation*}
\int \varphi(z-z') b^r_{z'} c_{z'} c_z \ud z'
	= \int (z-z')^\mu \varphi(z-z') b^r_{z'} c_{z'} \int_0^1 \ud t \,\nabla^\mu c_{z'+t(z-z')} \ud z'.
\end{equation*}
Thus, by \Cref{eqn.bdd.Facc.b(F)}, we conclude in a similar way the bound 
\begin{equation}
\label{eqn.bdd.phi*bcc}
\norm{\int \varphi(z-z') b^r_{z'} c_{z'} c_z \ud z'}
\lesssim 
k_F^{4+3/2} \norm{|\cdot|\varphi}_{L^1} + k_F^{4} \norm{|\cdot|\varphi}_{L^2}
\leq C k_F^{3/2} (ak_F)^{5/2} .
\end{equation}
\end{remark}

\begin{proof}[{Proof of \Cref{lem.bdd.F*ac.b(F)}}]
The main ingredient in the proof is a result of Birman and Solomjak \cite{Birman.Solomjak.1969,Birman.Solomjak.1980}, stated in \cite[Theorem 4.5]{Simon.2005}.
We start with \Cref{eqn.bdd.Fac.b(F)} and write the integral as 
\begin{equation*}
\int F(x-y) a_y c_y \ud y = \iint F^{x}(y) v(y-z) a_y a_z \ud y \ud z, \qquad F^{x}(y) = F(x-y).
\end{equation*}
We write the operator $K^{x}$ with integral kernel $K^{x}(y,z) = F^{x}(y) v(y-z)$ in terms of  its singular value decomposition.
That is, for some orthonormal systems $\left\{\phi^{x}_i\right\}$ and $\left\{\psi^{x}_i\right\}$ we have 
\begin{equation}\label{decF}
F^{x}(y) v(y-z) = \sum_i \mu_i \phi^{x}_i(y) \psi^{x}_i(z)
\end{equation}
with $\mu_i \geq 0$ the singular values. 
(Note that the singular values do not depend on $x$ since $K^x$ is unitarily equivalent to the translated operator with $x=0$; moreover, $\phi_i^x(y) = \phi_i^0(y-x)$ and likewise for $\psi^x_i$.) Thus 
\begin{equation*}
\int F(x-y) a_y c_y \ud y = \sum_{i} \mu_i \iint \phi^{x}_i(y) \psi^{x}_i(z)a_y a_z \ud y \ud z 
	= \sum_{i} \mu_i a(\overline{\phi_i^x}) a(\overline{\psi_i^x}).
\end{equation*}
In norm this is thus bounded by $\sum \mu_i = \norm{K^{x}}_{\mathfrak{S}_1}$, the trace-norm of $K^x$.
Furthermore, $K^x$ is unitarily equivalent to the dilated operator with kernel $k_F^{-3} F^x(k_F^{-1} y) v(k_F^{-1}(y-z))$. 
Since the functions $F^x(k_F^{-1}\cdot)$ and $k_F^{-3} \widehat{v(k_F^{-1}\cdot)} = \widehat{v}(k_F\cdot)$ are compactly supported (in configuration and momentum-space,  respectively) 
with diameter of support of order $1$, we can apply \cite[Theorem 4.5]{Simon.2005}\footnote{\cite[Theorem 4.5]{Simon.2005} is stated only for Euclidean space
but can be extended to the torus $\Lambda$ without much difficulty.}
and conclude that 
\begin{equation*}
\norm{\int F(x-y) a_y c_y}
\leq 
\sum_{i} \mu_i 
=\norm{K^x}_{\mathfrak{S}_1}
\leq C \norm{F^x(k_F^{-1} \cdot)}_{L^2} \norm{\widehat{v}(k_F\cdot)}_{\ell^2} \leq C k_F^{3/2} \norm{F}_{L^2},
\end{equation*}
where we introduced the notation $\|f\|_{\ell^2} = L^{-3/2} ( \sum_k |f(k)|^2)^{1/2}$. 

In order to prove \Cref{eqn.bdd.Facc.b(F)} we similarly write 
\begin{equation*}
\int F(x-y) a_y c_y \nabla^n c_{ty+(1-t)x} \ud y
= \iiint F^x(y) v(y-z) \nabla^n v^{(1-t)x}(ty-w) a_y a_z a_w \ud y \ud z \ud w
\end{equation*}
with $v^x(y)  = v(y+x)$. Using again \eqref{decF}, 
\begin{equation*}
\int F(x-y) a_y c_y \nabla^n c_{ty+(1-t)x} \ud y
= - \sum_i \mu_i  a(\overline{\psi^{x}_i}) \iint \phi^{x}_i(y) \nabla^n v^{(1-t)x}(ty-w) a_y a_w \ud y \ud w.
\end{equation*}
We do one more singular value decomposition, this time for the operators $\widetilde K_i^{x,t}$  with kernels 
$\widetilde K_i^{x,t}(y,w) = \phi^{x}_i(y) \nabla^n v^{(1-t)x}(ty-w) = \sum_j \widetilde\mu_{ij}^{x,t} \widetilde \phi_{ij}^{x,t}(y) \widetilde \psi_{ij}^{x,t}(w)$. 
In combination we thus obtain the bound 
\begin{equation*}
\norm{\int F(x-y) a_y c_y \nabla^n c_{ty+(1-t)x} \ud y} \leq \sum_i \mu_i \sum_j \widetilde \mu_{ij}^{x,t} .
\end{equation*}
Similarly as above, observe that $\widetilde K_i^{x,t}$ has the same trace norm as the dilated operator with integral kernel $ \phi^{x}_i(k_F^{-1} y) k_F^{-3} t^{3/2}\nabla^n v^{(1-t)x}(tk_F^{-1} (y-w))$. 
Each $\phi^{x}_i$ has compact support of diameter $\lesssim k_F^{-1}$ (since $F^x$ has). 
Furthermore, the Fourier transform of $k_F^{-3} t^{3/2}\nabla^n v^{(1-t)x}(tk_F^{-1} \cdot)$ has compact support  of diameter $t \leq 1$. 
Thus we can again apply \cite[Theorem 4.5]{Simon.2005}
to conclude 
\begin{equation*}
\sum_j \widetilde \mu_{ij}^{x,t} = 
\norm{\widetilde K_i^{x,t}}_{\mathfrak{S}_1} \leq C \norm{\phi^{x}_i(k_F^{-1}\cdot)}_{L^2} \norm{t^{-3/2} (|\cdot|^n\widehat{v})(t^{-1} k_F\cdot)}_{\ell^2}
	\leq C k_F^{3/2+n}
\end{equation*}
uniformly in $i$, $x$ and $t$. 
In combination with the bound $\sum_i \mu_i \leq C k_F^{3/2} \norm{F}_{L^2}$ from above, this concludes the proof of \Cref{eqn.bdd.Facc.b(F)}.
\end{proof}

In order to state some of the intermediary bounds in the proof of~\Cref{lemma.bdd.list.operators}, 
we introduce the following operators. 
\begin{defn}[Highly excited particles]\label{def:na}
For  $\alpha > 0$ we define 
\begin{equation*}%\label{def:na}
\mcN_{>\alpha} = \sum_{|k| > k_F(ak_F)^{-\alpha}} a_k^* a_k,
\qquad 
\mcN_> = \sum_{|k| > 2k_F} a_k^* a_k.
\end{equation*}
\end{defn}

\begin{remark}[{see also \cite[Proposition 4.15]{Giacomelli.2023}}]
For any $|k| > 2k_F$ we have $||k|^2-k_F^2|\geq k_F^2$, and 
for any $|k| > k_F(ak_F)^{-\alpha}$ with $\alpha > 0$ we have $||k|^2-k_F^2|\geq C k_F^2(ak_F)^{-2\alpha}$ for small $ak_F$.
Hence 
\begin{equation}\label{eqn.bdd.N>.a.priori}
\mcN_> 
\leq C k_F^{-2} \H_0,
\qquad
\mcN_{>\alpha}  
\leq C k_F^{-2}(ak_F)^{2\alpha} \H_0.
\end{equation}
\end{remark}

Before giving the proof of \Cref{lemma.bdd.list.operators} 
we first explain some intuition behind the proof.

\begin{remark}[Intuition behind the bounds]
Conceptually, the main difficulty is ``how to deal with the $b_x$'s''.
The operator-valued distribution $b_x$ is \emph{not} a bounded operator, so we cannot bound it in norm. 
Essentially we have $3$ different methods of treating a factor $b_x$. 

\begin{itemize}
\item We can bound one factor $b_x$ in terms of $\mcN$  via a computation similar to 
\begin{equation*}
\int \norm{b_x\xi} \ud x 
  \leq L^{3/2} \left(\int \norm{b_x\xi}^2 \ud x\right) 
  \leq L^{3/2} \expect{\mcN}_{\xi}^{1/2}
  \leq C N^{1/2} k_F^{-3/2} \expect{\mcN}_\xi^{1/2}.
\end{equation*}
Analogously, the regularized operator $b_x^r$ can be bounded in terms of  $\mcN_>$
and the operator $b_x^>$ (introduced in \Cref{eqn.def.b>} below) can be bounded in terms of $\mcN_{>\alpha}$.

\item We can bound two factors of $b_x$ in terms of $\Q_4$ using a factor $V(x-y)$ via a computation similar to 
\begin{equation*}
\begin{aligned}
\iint V(x-y) \norm{b_x b_y \xi} \ud x \ud y 
  & 
    \leq \left(\iint V(x-y)\ud x \ud y\right)^{1/2} 
    \left(\iint V(x-y) \norm{b_x b_y \xi}^2 \ud x \ud y\right)^{1/2}
  \\ & 
    = L^{3/2} \norm{V}_{L^1}^{1/2} \expect{2\Q_4}_\xi^{1/2}
    \leq C N^{1/2} a^{1/2} k_F^{-3/2} \expect{\Q_4}_\xi^{1/2}.  
\end{aligned}
\end{equation*}

\item Finally, we can bound one factor $b_x$ ``in norm'' using any appropriate (compactly supported) 
function $F$ and a factor $c_x$ (or two) via an application of \Cref{lem.bdd.F*ac.b(F)}.
\end{itemize} 

Factors of $c_x$ are not problematic. These we either bound in norm as $\norm{c_x}\leq C k_F^{3/2}$ or Taylor expand 
and bound the resulting $\nabla c_x$ in norm as $\norm{\nabla c_x}\leq C k_F^{3/2+1}$. 
Similarly, factors of $v(x-y)$ can either be bounded by $\norm{v}_{L^2}\leq C k_F^{3/2}$ or $\norm{v}_{L^\infty}\leq C k_F^3$. 
Finally, for factors $u^r(x-y)$ we either use $0\leq u^r\leq \mathbbm{1}$ as operators 
or write $u^r(x-y) = \delta(x-y) - v^r(x-y)$ and compute separately the terms with $\delta$ and $v^r$.
Here we note that $\norm{v^r}_{L^1} \leq C$ by \Cref{lem.bdd.vr.L1}. 
\end{remark}

The remainder of this section gives the proof of \Cref{lemma.bdd.list.operators}.
We prove each of the bounds separately.

\subsection{Taylor-expansion errors}
We bound all error-terms with a superscript ``Taylor'' in \Cref{eqn.[H0.B]+Q2.decompose}, resulting from the remainder term in the expansion \eqref{eqn.expand.cx.2nd.order}. 
Together they are of the form 
\begin{equation}\label{eqn.formula.T.Taylor}
\T 
=
\iint F^{\mu\nu}(x-y) b_x b_y c_y \int_0^1 \ud t \, (1-t) \nabla^\mu \nabla^\nu c_{x + t(y-x)} \ud x \ud y + \hc
\end{equation}
for the function(s) (using \Cref{eqn.defn.mcE_varphi.scat.eqn})
\begin{equation*}
F^{\mu\nu}(x) 
= - \left[x^\mu x^\nu \Delta\varphi(x) + \frac{1}{2}  x^\mu x^\nu V(x) (1 - \varphi(x))\right]
= 2 x^\mu \nabla^\nu \varphi(x) - x^\mu \mcE_\varphi^\nu(x).
\end{equation*}
We note that by \Cref{lem.prop.phi} and \Cref{eqn.bdd.mcE_phi.L1} we have
(with $\norm{F}_{L^1} = \sum_{\mu,\nu} \norm{F^{\mu\nu}}_{L^1}$ and similarly for the $L^2$ norm)
\begin{equation*}
\norm{F}_{L^1} \leq C a^3 \abs{\log ak_F},
\qquad 
\norm{F}_{L^2} \leq C a^{3/2}.
\end{equation*}
To bound $\expect{\T}_\xi$ we write 
\begin{equation}\label{eqn.def.b>}
b_x = b_x^< + b_x^>, \qquad 
  % \qquad 
  b_x^> = \frac{1}{L^{3/2}} \sum_{|k| > k_F(ak_F)^{-\alpha}} e^{ikx} a_k
\end{equation}
for some $\alpha > 0$. % (to be optimised later). 
Note that 
$\mcN_{>\alpha} 
= \int (b_x^>)^* b_x^> \ud x$ with $\mcN_{>\alpha}$ from \Cref{def:na}. 
For the term with $b_x^<$ 
we bound $\norm{b_x^<} \leq C k_F^{3/2} (ak_F)^{-3\alpha/2}$ and 
$\norm{\nabla^\mu \nabla^\nu c_{x + t(y-x)}}\leq C k_F^{3/2+2}$.
Then by Cauchy--Schwarz
\begin{align*}
& \abs{\expect{\iint F^{\mu\nu}(x-y) b_x^< b_y c_y \int_0^1 \ud t \, (1-t) \nabla^\mu \nabla^\nu c_{x + t(y-x)} \ud x \ud y}_\xi  }
\\ & \quad \leq 
C \norm{F}_{L^1} k_F^{5+3/2} (ak_F)^{-3\alpha/2}  \int \norm{b_y \xi}\ud y
% \\ & \quad 
\leq 
C N^{1/2} a^3 k_F^{5} \abs{\log ak_F} (ak_F)^{-3\alpha/2}   \expect{\mcN}_\xi^{1/2}.
\end{align*}
Here we used that $b_x$ and $c_y$ anti-commute to reorder the annihilation operators 
such that $b_y$ is the last operator and  we can bound all others in norm, while $b_y$ yields the bound $\norm{b_y\xi}$.
For the terms with $b_x^>$ we use \Cref{eqn.bdd.Facc.b(F)} to bound the $y$-integral. 
To do this we write $b_y = a_y - c_{-y}$. The term with $c_{-y}$ can be bounded in the same way as the term with $b_x^<$ above. 
For the term with $a_y$ we obtain 
\begin{equation*}
\begin{aligned}
& \abs{\expect{\iint F^{\mu\nu}(x-y) b_x^> a_y c_y \int_0^1 \ud t \, (1-t) \nabla^\mu \nabla^\nu c_{x + t(y-x)} \ud x\ud y}_\xi  }
\\ & \quad \leq 
C k_F^{5} \norm{F}_{L^2} \int \norm{b_x^> \xi} \ud x
% \\ & \quad 
\leq 
C N^{1/2} a k_F^3 (ak_F)^{1/2} \expect{\mcN_{>\alpha}}_\xi^{1/2}.
\end{aligned}
\end{equation*}
% \end{align*}
We conclude the bound 
\begin{equation*}
\begin{aligned}
\abs{\expect{\T}_\xi}
  & \lesssim 
  N^{1/2} a^3 k_F^{5} \abs{\log ak_F} (ak_F)^{-3\alpha/2}   \expect{\mcN}_\xi^{1/2}
  % \\ & \quad 
  + 
  N^{1/2} a k_F^3 (ak_F)^{1/2}  \expect{\mcN_{>\alpha}}_\xi^{1/2}
\end{aligned}
\end{equation*}
for any $\alpha > 0$. Using the bound on $\mcN_{>\alpha}$ in \Cref{eqn.bdd.N>.a.priori} we conclude the desired.

\subsection{Scattering equation cancellation}
We proceed with bounding $\Q_{\textnormal{scat}}$  in \Cref{eqn.def.Qscat}.
This bound is analogous to that of $\T$ above, only $F^{\mu\nu}$ is replaced by $\mcE_\varphi^\mu$ and $\nabla^\mu \nabla^\nu c$ is replaced by $\nabla^\mu c$. 
That is, the bound for $\Q_{\textnormal{scat}}$ is as for $\T$, only we have one fewer power of $k_F$ (from $\norm{\nabla c}$ instead of $\norm{\nabla^2 c}$)
and norms of $F$ are replaced by norms of $\mcE_\varphi$. 
We conclude the bound 
\begin{equation*}
\begin{aligned}
\abs{\expect{\Q_{\textnormal{scat}}}_\xi} 
	& \leq C N^{1/2} k_F^4 (ak_F)^{-3\alpha/2} \norm{\mcE_\varphi}_{L^1} \expect{\mcN}_\xi^{1/2}
		+ C N^{1/2} k_F^{3/2+1} \norm{\mcE_\varphi}_{L^2} \expect{\mcN_{>\alpha}}_\xi^{1/2}
  \\ & 
  \leq C N^{1/2} a^3 k_F^5 (ak_F)^{-3\alpha/2} \abs{\log ak_F}   \expect{\mcN}_\xi^{1/2} 
  + C N^{1/2} a k_F^{3} (ak_F)^{1/2} \expect{\mcN_{>\alpha}}_\xi^{1/2}
\end{aligned}
\end{equation*}
for any $\alpha > 0$, where we have used \Cref{lem.prop.phi} and \Cref{eqn.bdd.mcE_phi.L1} for the bound on the  norms of $\mcE_\varphi$. 
Again, using the bound on $\mcN_{>\alpha}$ in \Cref{eqn.bdd.N>.a.priori} we conclude the desired.

\subsection{(Non-)regularization of \texorpdfstring{$[\H_0,B]$}{[H0,B]}}
Next we bound $\H_{0;B}^{\div r}$ defined in \Cref{eqn.[H0.B].ini}, given by 
\begin{equation}\label{formula.H.non-reg}
\H_{0;B}^{\div r} 
  = \iint \left(\Delta \varphi(x-y) c_y c_x  + 2 \nabla^\mu \varphi(x-y) c_y \nabla^\mu c_x\right) 
    \left(b_x^r b_y^r - b_x b_y \right) \ud x \ud y + \hc.
\end{equation}
To bound this, write $b_x = b_x^r + c_x^r$ with the latter operator satisfying $\norm{c_x^r}\leq C k_F^{3/2}$ since 
$b_k$ and $b_k^r$ agree for all momenta $|k| > 3k_F$.
Further, we Taylor expand the factor $c_y$ in the first term as in \Cref{eqn.Taylor.c_x.1st.order}.
Then we have (changing variables to $z=y-x$)
\begin{equation}\label{eqn.H0B.two.terms}
\begin{aligned}
\expect{\H_{0;B}^{\div r}}_\xi 
	& = - 2 \Re \iint z^\mu \Delta \varphi(z) \int_0^1 \ud t \expect{\nabla^\mu c_{x+tz} c_x (c_x^r b_{x+z}^r + b_x^r c_{x+z}^r + c_x^r c_{x+z}^r)}_\xi \ud x \ud z 
	\\ & \quad 
		+ 2 \Re \iint \nabla^\mu \varphi(z) \expect{c_{x+z} \nabla^\mu c_x (c_x^r b_{x+z}^r + b_x^r c_{x+z}^r + c_x^r c_{x+z}^r)}_\xi \ud x \ud z.
\end{aligned}
\end{equation}
The two terms are treated similarly. We start with the first. 
Define for any state $\xi$ the function $\phi(x,z) = \expect{\nabla^\mu c_{x+tz} c_x (c_x^r b_{x+z}^r + b_x^r c_{x+z}^r + c_x^r c_{x+z}^r)}_\xi$
and note that it vanishes at $z=0$. Thus we Taylor expand: (with $\nabla^\mu_2$ denoting a derivative in the second argument)
\begin{equation*}
\phi(x,z) = z^\nu \int_0^1 \ud s \, [\nabla^\nu_2 \phi](x,sz).
\end{equation*}
The derivative hits either a factor $c$ or $c^r$ or a factor $b^r$. 
In the terms where the derivative hits a factor $c$ or $c^r$  we bound the $c$'s and $c^r$'s in norm. 
(Apart from the term with only $c$'s and $c^r$'s, where we keep one factor $c^r$ without a derivative.)
Recall that $\norm{\nabla^n c} \leq C k_F^{3/2+n}$ for any $n$. Similarly $\norm{\nabla^n c_x^r} \leq C k_F^{3/2+n}$.
Thus, the terms where the derivative hits a factor $c$ or $c^r$ are bounded by 
\begin{equation*}
k_F^{3/2+5} \norm{|\cdot|^2 \Delta \varphi}_{L^1} \int \left(\norm{b_x^r\xi} + \norm{c_x^r\xi}\right) \ud x 
\leq C N^{1/2} a^3 k_F^5 \abs{\log ak_F} \expect{\mcN}_\xi^{1/2}
\end{equation*}
using the bounds of \Cref{lem.prop.phi}. 
For the term where the derivative hits a factor $b^r$ we again bound the $c$'s and $c^r$'s in norm. 
These terms are then bounded by 
\begin{equation*}
k_F^{3/2+4} \norm{|\cdot|^2 \Delta \varphi}_{L^1} \int \norm{\nabla b_x^r\xi} \ud x 
\leq C N^{1/2} a^3 k_F^4 \abs{\log ak_F} \expect{\H_0}_\xi^{1/2}
+ C N^{1/2} a^3 k_F^5 \abs{\log ak_F} \expect{\mcN}_\xi^{1/2}
\end{equation*}
since $\int \norm{\nabla b_x^r \xi}^2 \ud x \leq \expect{\H_0}_\xi + k_F^2 \expect{\mcN}_\xi $.
The second term in \Cref{eqn.H0B.two.terms} is treated analogously, only the bound has a factor $\norm{|\cdot|\nabla \varphi}_{L^1}$ 
instead of $\norm{|\cdot|^2 \Delta\varphi}_{L^1}$. 
Collecting all the terms, we conclude the desired bound on $\H_{0;B}^{\div r}$.

\subsection{Error terms from \texorpdfstring{$[\Q_4,B]$}{[Q4,B]}}
We now bound $\Q_{4;B}^\mcE$, given in \Cref{eqn.def.[Q4.B]-error}. 
We split the operator into $4$ terms and bound each separately: 
\begin{enumerate}[1.]
\item  The quartic term with one $\delta$ and one $v^r$, 
\item The quartic term with two $v^r$'s, 
\item The order $6$ term with $\delta$, and 
\item The order $6$ term with $v^r$.
\end{enumerate}

\subsubsection{Quartic term with one $\delta$ and one $v^r$:}
This term is of the form 
\begin{equation*}
\A_1 = \iiint \ud x \ud y \ud z' \, V(x-y) \varphi(y-z') v^r(z'-x) b_y b_x c_{z'} c_y + \hc.
\end{equation*}
We bound $\norm{c_{z'}c_y}\leq C|y-z'| k_F^4$ and $\norm{v^r}_{L^\infty}\leq C k_F^3$.
Then by Cauchy--Schwarz we get  
(using \Cref{lem.prop.phi} to bound the norm of $\varphi$)
\begin{equation*}
\begin{aligned}
\abs{\expect{\A_1}_\xi}
 & 
  \leq 
  C k_F^7 \norm{ |\cdot|\varphi}_{L^1} \iint V(x-y) \norm{b_x b_y \xi} \ud x \ud y 
  \leq  C k_F^7 \norm{|\cdot|\varphi}_{L^1} \norm{V}_{L^1}^{1/2} L^{3/2} \expect{\Q_4}_{\xi}^{1/2} 
\\ & 
  \leq C N^{1/2} a^{3+1/2} k_F^{4+1/2}  \expect{\Q_4}_\xi^{1/2}.
\end{aligned}
\end{equation*}
(Here we reordered the annihilation operators such that $b_x b_y$ is  last so we can bound $c_{z'}c_y$ in norm and get the factor $\norm{b_xb_y\xi}$.)

\subsubsection{Quartic term with two $v^r$'s:}
This term is of the form 
\begin{equation*}
\A_2 = -\frac{1}{2}\iiiint \ud x \ud y \ud z \ud z' \, V(x-y) \varphi(z-z') v^r(z'-x) v^r(z-y)  b_y b_x c_{z'} c_z + \hc.
\end{equation*}
As above, we bound $\norm{c_{z'}c_z}\leq C |z-z'| k_F^4$. 
By Cauchy--Schwarz the $z,z'$ integrals are then bounded as 
\begin{equation*}
\iint \ud z \ud z' \, |\varphi(z-z')||z-z'| |v^r(z'-x)| |v^r(z-y)|  
\leq \norm{|\cdot|\varphi}_{L^1} \norm{v^r}_{L^2}^2 \leq C a^3 k_F^2.
\end{equation*}
Then by Cauchy--Schwarz as above we obtain 
\begin{equation*}
\abs{\expect{\A_2}_\xi}
  \leq C a^3 k_F^6 \iint V(x-y) \norm{b_x b_y \xi} \ud x \ud y 
  \leq C N^{1/2} a^{3+1/2} k_F^{4+1/2} \expect{\Q_4}_\xi^{1/2}. 
\end{equation*}

\subsubsection{Order $6$ term with $\delta$:}
This term is of the form 
\begin{equation*}
\A_3 = -\iiint \ud x \ud y \ud z' \, V(x-y) \varphi(x-z') b_y^* b_{z'}^r b_y b_x c_{z'} c_{x} + \hc.
\end{equation*}
Computing the expectation in some state $\xi$ we bound the $z'$-integral as 
\begin{equation*}
\abs{\int \ud z' \, \varphi(x-z') \expect{b_y^* b_{z'}^r b_y b_x c_{z'} c_x}_\xi }
\leq  \norm{\int \ud z' \, \varphi(x-z') b_{z'}^r c_{z'} c_{x}} \norm{b_y \xi} \norm{b_x b_y \xi}.
\end{equation*}
Using \Cref{eqn.bdd.phi*bcc} 
to bound the first factor we get by Cauchy--Schwarz
\begin{equation*}
\begin{aligned}
\abs{\expect{\A_3}_\xi}
  & \leq 
  C (ak_F)^{5/2} k_F^{3/2}  \iint V(x-y) \norm{b_y \xi} \norm{b_xb_y\xi} \ud x \ud y 
  \\
  & \leq 
  C (ak_F)^{5/2} k_F^{3/2}  \norm{V}_{L^1}^{1/2} \expect{\mcN}_\xi^{1/2} \expect{\Q_4}_{\xi}^{1/2}
  \\
  & \leq 
  C a^3 k_F^4 \expect{\mcN}_\xi^{1/2} \expect{\Q_4}_\xi^{1/2}.
\end{aligned}   
\end{equation*}

\subsubsection{Order $6$ term with $v^r$:}
This term is of the form 
\begin{equation*}
\A_4 = -\iiiint \ud x \ud y \ud z \ud z' \, V(x-y) \varphi(z-z') v^r(z-x) b_y^* b_{z'}^r b_y b_x c_{z'} c_z + \hc.
\end{equation*}
Computing the expectation in some state $\xi$ we bound the $z,z'$-integrals, again using \Cref{eqn.bdd.phi*bcc}, as % with $b^r$ instead of $b$, as
\begin{equation*}
\begin{aligned}
& \abs{\iint \ud z \ud z' \, \varphi(z-z') v^r(z-x) \expect{b_y^* b_{z'}^r b_y b_x c_{z'} c_z}_\xi}
\\ & \quad 
  \leq \int \ud z \, |v^r(z-x)| \norm{ \int \ud z' \, \varphi(z-z') b_{z'}^r  c_{z'}c_z} \norm{b_y \xi} \norm{b_x b_y\xi}
\\ & \quad 
  \leq C (ak_F)^{5/2} k_F^{3/2} \norm{v^r}_{L^1} \norm{b_y \xi} \norm{b_x b_y\xi}.
\end{aligned}
\end{equation*}
Recall that $\norm{v^r}_{L^1} \leq C$ by \Cref{lem.bdd.vr.L1}. 
Using then Cauchy--Schwarz as above we find 
\begin{equation*}
\abs{\expect{\A_4}_\xi}
  \leq C a^3 k_F^4 \expect{\mcN}_\xi^{1/2} \expect{\Q_4}_\xi^{1/2}.
\end{equation*}
This gives the desired bound on $\Q_{4;B}^\mcE$.

\subsection{Error terms from \texorpdfstring{$[\Q_2,B]$}{[Q2,B]}}
Finally we bound $\Q_{2;B}^\mcE$. Recalling \Cref{eqn.com.bbcc.b*b*c*c*,eqn.calc.[Q2.B]}, it is given by 
\begin{equation*}
\begin{aligned}
\Q_{2;B}^\mcE
  & = \frac{1}{4} \iiiint \ud x \ud y \ud z \ud z' \, 
    V(x-y) \varphi(z-z') 
    \Bigl\{
		b_{x}^* b_{y}^* b_{z}^r b_{z'}^r [c_{y}^* c_x^* ,  c_{z'} c_z]
    +  c_{y}^* c_x^* c_{z'} c_z [b_{x}^* b_{y}^*  , b_{z}^r b_{z'}^r] 
  \\ & \qquad
    + 
    \Bigl(
    (\delta(x-z)\delta(y-z')  - \delta(x-z') \delta(y-z))
    (v(x-z)v(y-z')  - v(x-z') v(y-z))
    \\ & \qquad \quad 
    -[b_{x}^* b_{y}^*  , b_{z}^r b_{z'}^r ] [c_{y}^* c_x^*, c_{z'} c_z]
    \Bigl)
    \Bigr\}
    + \hc .
\end{aligned}
\end{equation*}
We split the terms into three groups: 
\begin{enumerate}[1.]
\item The terms with only the $c$-commutator (i.e. of the form $b^* b^* b b [c^*c^*,cc]$),
\item The terms with only the $b$-commutator, and 
\item The rest, involving terms with both commutators, with the leading term [of the form $\delta\delta v v$] subtracted.
\end{enumerate}

\subsubsection{Terms with only the $c$-commutator:}
Recalling the formula for the commutator in \Cref{eqn.calc.[cc.cc]} all terms with $c^*c$ give the same contribution, 
so do all the constant terms.
That is, we need to bound terms of the form 
\begin{equation*}
\begin{aligned}
\A_{1a} & = -\iiiint V(x-y) \varphi(z-z') v(x-z) b_x^* b_y^* b_z^r b_{z'}^r c_y^* c_{z'} \ud x \ud y \ud z \ud z' + \hc,
\\
\A_{1b} & = \frac{1}{2}\iiiint V(x-y) \varphi(z-z') v(x-z) v(y-z') b_x^* b_y^* b_z^r b_{z'}^r \ud x \ud y \ud z \ud z' + \hc.
\end{aligned}
\end{equation*}
For $\A_{1a}$ we note that as an operator $0 \leq v \leq \mathbbm{1}$. 
Then by Cauchy--Schwarz we have for any $\eps > 0$
(recall that $\{b_x, c_y^*\} = 0$)
\begin{equation*}
\begin{aligned}
\pm\A_{1a} & = \mp\iint \ud x \ud z \, v(x-z) 
    \left[\int \ud y \, V(x-y)  b_{x}^* b_{y}^*  c_y^* \right]
    \left[\int \ud z'\, \varphi(z-z') b_{z}^r b_{z'}^r c_{z'} \right]
    + \hc 
  \\ & \leq 
    \eps k_F^{-2} \iiint \ud x \ud y \ud z \, V(x-y) V(x-z) b_{x}^* b_{z}^*  c_z^* c_y b_y b_x 
  \\ & \quad 
    + \eps^{-1} k_F^2 \iiint \ud x \ud y \ud z \, \varphi(x-y) \varphi(x-z) c_z^* (b_{z}^r)^* (b_{x}^r)^*   b_{x}^r b_{y}^r  c_y
  \\ & =: \eps \A_{1a}^V + \eps^{-1} \A_{1a}^\varphi.
\end{aligned}
\end{equation*}
For $\A_{1a}^V$ we bound $c_z^*,c_y$ in norm. Then by Cauchy--Schwarz we have
\begin{equation*}
\begin{aligned}
\abs{\expect{\A_{1a}^V}_\xi}
  & \leq C k_F \iiint V(x-y)V(x-z) \norm{b_z b_x \xi} \norm{b_y b_x \xi} \ud x \ud y \ud z
  \\
  & \leq C k_F \iiint V(x-y)V(x-z) 
  \norm{b_y b_x \xi}^2
  \ud x \ud y \ud z 
  \leq C ak_F \expect{\Q_4}_\xi.
\end{aligned}
\end{equation*}
For $\A_{1a}^\varphi$ we use \Cref{eqn.bdd.phi*bc} 
to bound the $y$- and $z$-integrations. 
We find 
\begin{equation*}
\abs{\expect{\A_{1a}^\varphi}_\xi }
  \leq k_F^2 \int \ud x \norm{\int \ud y\, \varphi(x-y) b_y^r c_y}^2 
  \norm{b_x^r\xi}^2
  % \expect{(b_x^r)^* b_x^r}_\xi 
  \leq C k_F^2 (ak_F)^{3} \expect{\mcN_>}_\xi.
\end{equation*}
Optimising in $\eps$ we thus obtain the bound  
\begin{equation}\label{eqn.Q2B.A.1.1}
\abs{\expect{\A_{1a}}_\xi}
  \leq C a^2 k_F^3 \expect{\mcN_>}_\xi^{1/2} \expect{\Q_4}_{\xi}^{1/2}.
\end{equation}

 Next, we bound $\A_{1b}$ as
\begin{equation*}
\abs{\expect{\A_{1b}}_\xi}
	\leq \iiint \ud x \ud y \ud z \, V(x-y) |v(x-z)| \norm{\int \ud z' \, \varphi(z-z') v(y-z') b_{z'}^r} \norm{b_z^r \xi} \norm{b_yb_x \xi}.
\end{equation*}
Since $0\leq \hat u^r \leq 1$, we have
\begin{equation*}
\norm{\int \ud z' \, \varphi(z-z') v(y-z') b_{z'}^r} \leq  \norm{\int \ud z' \, \varphi(z-z') v(y-z') a_{z'}} = \left( \varphi^2 * v^2(y-z) \right)^{1/2}
\end{equation*}
(with $*$ denoting convolution) 
and hence, by Cauchy--Schwarz,
\begin{align}
\abs{\expect{\A_{1b}}_\xi}
	& \leq \left[\iiint \ud x \ud y \ud z \, V(x-y) |v(x-z)|^2 \norm{ b_y b_x \xi}^2\right]^{1/2}
	\nonumber
	\\ & \quad \times 
	\left[\iiint \ud x \ud y \ud z \,   V(x-y) \varphi^2 * v^2(y-z) \norm{b_z^r \xi}^2 \right]^{1/2}
	\nonumber
	\\ & \leq 
	 \norm{v}_{L^2}^2 \expect{\Q_4}_\xi^{1/2} \norm{V}_{L^1}^{1/2} \|\varphi\|_{L^2}\expect{\mcN_>}_\xi^{1/2} 
	\leq C a^{2} k_F^{3}  \expect{\Q_4}_\xi^{1/2}  \expect{\mcN_>}_\xi^{1/2}  ,
\label{eqn.Q2B.A.1.2}
\end{align}
where we have used Lemma~\ref{lem.prop.phi} to bound the $L^2$-norm of $\varphi$ in the final step.

\subsubsection{Terms with only the $b$-commutator:}
Recalling the formula for the commutator in \Cref{eqn.calc.[bb.bb]} all terms with $b^*b$ give the same contribution, 
so do all the constant terms.
That is, we need to bound terms of the forms 
\begin{equation*}
\begin{aligned}
\A_{2a} & = -\iiiint V(x-y) \varphi(z-z') u^r(x-z) b_y^* b_{z'}^r c_y^* c_x^* c_{z'} c_z \ud x \ud y \ud z \ud z' + \hc,
\\
\A_{2b} & = \frac{1}{2}\iiiint V(x-y) \varphi(z-z') u^r(x-z) u^r(y-z')  c_y^* c_x^* c_{z'} c_z \ud x \ud y \ud z \ud z'  + \hc.
\end{aligned}
\end{equation*}
For $\A_{2a}$ we use that $0 \leq u^r \leq \mathbbm{1}$ as an operator.
Thus, for any $\eps > 0$, 
\begin{equation*}
\begin{aligned}
\pm\A_{2a}
  & = \mp\iint \ud x \ud z \, u^r(x-z) 
    \left[\int \ud y \, V(x-y) b_y^* c_{y}^* c_x^* \right]
  	\left[\int \ud z' \, \varphi(z-z')  b_{z'}^r c_{z'} c_z \right]
    + \hc 
  \\
  & \leq 
    \eps k_F^{-2} \iiint \ud x \ud y \ud z \,   V(x-y) V(x-z) b_z^* c_z^* c_x^*  c_x c_y b_y 
  \\ & \quad 
    + \eps^{-1} k_F^2 \iiint \ud x \ud y \ud z \,   \varphi(x-y) \varphi(x-z) c_x^* c_z^* (b_z^r)^* b_y^r c_{y} c_x
  \\ & =: \eps \A_{2a}^V + \eps^{-1} \A_{2a}^\varphi.
\end{aligned}
\end{equation*}
For $\A_{2a}^V$ we bound $\norm{c_xc_y}\leq C |x-y|k_F^4$. Then by Cauchy--Schwarz we have 
\begin{equation*}
\begin{aligned}
\abs{\expect{\A_{2a}^V}_\xi}
& \leq C k_F^6 \iiint \ud x \ud y \ud z \,   V(x-y) V(x-z) |x-y| |x-z| \norm{b_z\xi} \norm{b_y\xi}
\\
& 
\leq C k_F^6 \norm{|x|V}_{L^1}^2 \expect{\mcN}_\xi 
\leq C a^4 k_F^6 \expect{\mcN}_\xi.
\end{aligned}
\end{equation*}
Analogously
\begin{equation*}
\abs{\expect{\A_{2a}^\varphi}_\xi}
\leq C k_F^{10} \norm{|x|\varphi}_{L^1}^2 \expect{\mcN_>}_\xi 
\leq C a^6 k_F^8 \expect{\mcN_>}_\xi.
\end{equation*}
Optimising in $\eps$ we find 
\begin{equation}\label{eqn.Q2B.A.2.1}
\abs{\expect{\A_{2a}}_\xi}
  \leq C a^5 k_F^7 \expect{\mcN}_\xi^{1/2} \expect{\mcN_>}_\xi^{1/2}  \leq C a^5 k_F^7 \expect{\mcN}_\xi .
\end{equation}

For $\A_{2b}$ we write $u^r = \delta - v^r$. 
The term with both factors $\delta$ is 
\begin{equation*}
\A_{2b}^{\delta\delta} = -\iint V(x-y) \varphi(x-y) c_x^* c_y^* c_y c_x \ud x \ud y.
\end{equation*}
By Taylor expanding $c_y$ and $c_y^*$ as in \Cref{eqn.Taylor.c_x.1st.order} and bounding $\varphi \leq 1$ we have 
\begin{equation*}
\abs{\expect{\A_{2b}^{\delta\delta}}_\xi}
	\leq C k_F^5 \norm{|\cdot|^2 V}_{L^1} \expect{\mcN}_\xi \leq C a^3 k_F^5 \expect{\mcN}_\xi.
\end{equation*}
Both terms with one factor $\delta$ and one factor $v^r$ can be treated the same way. 
They are of the form 
\begin{equation*}
\A_{2b}^{v\delta} = \frac{1}{2}\iiint V(x-y) \varphi(z-y) v^r(x-z)  c_y^* c_x^* c_y c_z \ud x \ud y \ud z + \hc.
\end{equation*}
Thus, by Cauchy--Schwarz and Taylor-expanding  $c_y$ as in \Cref{eqn.Taylor.c_x.1st.order}
\begin{equation*}
\begin{aligned}
\abs{\expect{\A_{2b}^{v\delta}}_\xi}
	& \leq \left[\iiint V(x-y) v^r(x-z)^2  \norm{c_x c_y \xi}^2 \ud x \ud y \ud z\right]^{1/2}
	\\ & \quad \times 
	\left[\iiint V(x-y) \varphi(z-y)^2 \norm{c_y c_z \xi}^2 \ud x \ud y \ud z\right]^{1/2}
	\\ & \leq 
	C k_F^{5} \norm{|\cdot|^2V}_{L^1}^{1/2} \norm{v^r}_{L^2} \norm{V}_{L^1}^{1/2} \norm{|\cdot|\varphi}_{L^2} \expect{\mcN}_\xi
	\leq C a^4 k_F^6 (ak_F)^{1/2} \expect{\mcN}_\xi.
\end{aligned}
\end{equation*}
Finally, we bound the term $\A_{2b}^{vv}$ with two factors $v^r$ by Cauchy--Schwarz as
\begin{equation*}
\begin{aligned}
\abs{\expect{\A_{2b}^{vv}}_\xi}
	& \leq \left[\iiiint V(x-y) \varphi(z-z') v^r(x-z)^2  \norm{c_x c_y \xi}^2 \ud x \ud y \ud z \ud z'\right]^{1/2}
	\\ & \quad \times 
	\left[\iiiint V(x-y) \varphi(z-z') v^r(y-z')^2  \norm{c_{z'} c_{z} \xi}^2 \ud x \ud y \ud z \ud z'\right]^{1/2}
	\\ & 
	\leq C k_F^5 \norm{|\cdot|^2 V}_{L^1}^{1/2} \norm{\varphi}_{L^1}^{1/2} \norm{V}_{L^1}^{1/2} \norm{|\cdot|^2\varphi}_{L^1}^{1/2} \norm{v^r}_{L^2}^2 \expect{\mcN}_\xi 
	% \\ & 
	\leq C a^5 k_F^7 \abs{\log ak_F}^{1/2}\expect{\mcN}_\xi.
\end{aligned}	
\end{equation*}
We conclude that 
\begin{equation}\label{eqn.Q2B.A.2.2}
\abs{\expect{\A_{2b}}_\xi} \leq C a^3 k_F^5 \expect{\mcN}_\xi.
\end{equation}

\subsubsection{Terms with both commutators:}
Recalling the formulas for the commutators in \Cref{eqn.calc.[bb.bb],eqn.calc.[cc.cc]} 
we split the terms into four groups 
\begin{enumerate}[a.]
\item Terms of the form $b^*bc^*c$, 
\item Terms of the form $b^*b$,
\item Terms of the form $c^*c$, and 
\item The constant (i.e. fully contracted) terms.
\end{enumerate}

\paragraph{---Terms of the form $b^*bc^*c$:}
These terms have one factor $u^r$ and one factor $v$. 
The factors $u^r$ and $v$ can either have both different arguments, share one variable or have the same argument.
That is, we have the different types of terms
\begin{equation*}
\begin{aligned}
\A_{3a,1}
  & = -\frac{1}{4}\iiiint \ud x \ud y \ud z \ud z' \, V(x-y) \varphi(z-z') 
    u^r(x-z) v(y-z') 
    b_y^* c_x^* c_z b_{z'}^r + \hc,
\\
\A_{3a,2}
  & = \frac{1}{4}\iiiint \ud x \ud y \ud z \ud z' \, V(x-y) \varphi(z-z') 
    u^r(x-z) v(y-z) 
    b_y^* c_x^* c_{z'} b_{z'}^r + \hc,
\\
\A_{3a,3}
  & = -\frac{1}{4}\iiiint \ud x \ud y \ud z \ud z' \, V(x-y) \varphi(z-z') 
  u^r(x-z) v(x-z) 
  b_y^* c_y^* c_{z'} b_{z'}^r + \hc.
\end{aligned}
\end{equation*}
We consider the term $\A_{3a,1}$ and write $u^r(x-z) = \delta(x-z) - v^r(x-z)$.
For the term with $\delta$ we bound $\norm{v}_{L^\infty}\leq C k_F^3$ and $\norm{c_x}\leq C k_F^{3/2}$. 
Then by Cauchy--Schwarz the expectation value of  this term in a state $\xi$ is bounded by 
\begin{equation*}
\begin{aligned}
C k_F^6 \iiint V(x-y) \varphi(x-z') \norm{ b_y \xi} \norm{ b_{z'}^r \xi}  \ud x \ud y \ud z' 
% \\ 
& \leq C k_F^6 \norm{V}_{L^1} \norm{\varphi}_{L^1} \expect{\mcN}_\xi^{1/2} \expect{\mcN_>}_\xi^{1/2}
\\ & 
\leq C a^4 k_F^6 \abs{\log ak_F} \expect{\mcN}_\xi^{1/2} \expect{\mcN_>}_\xi^{1/2}.
\end{aligned}
\end{equation*}
Similarly we bound the term with $v^r$ by 
\begin{equation*}
\begin{aligned}
&
C k_F^3 \iiiint V(x-y) \varphi(z-z') |v^r(x-z)| |v(y-z')| \norm{ b_y \xi} \norm{ b_{z'}^r \xi}  \ud x \ud y \ud z \ud z' 
\\ 
& \quad 
\leq C k_F^3 
\left[\iiiint V(x-y) \varphi(z-z') v(y-z')^2 \norm{ b_y \xi}^2  \ud x \ud y \ud z \ud z'\right]^{1/2} 
\\ & \qquad \times 
\left[\iiiint V(x-y) \varphi(z-z') v^r(x-z)^2  \norm{ b_{z'}^r \xi}^2  \ud x \ud y \ud z \ud z'\right]^{1/2}	 
\\ & \quad 
\leq C a^4 k_F^6 \abs{\log ak_F} \expect{\mcN}_\xi^{1/2} \expect{\mcN_>}_\xi^{1/2}
\end{aligned}
\end{equation*}
using that $\norm{v}_{L^2}, \norm{v^r}_{L^2}\leq C k_F^{3/2}$.
The terms $\A_{3a,2}$ and $\A_{3a,3}$ can be bounded in the same way.
We conclude that
\begin{equation}\label{eqn.Q2B.A.3a}
\abs{\expect{\A_{3a,1}}_\xi} +
\abs{\expect{\A_{3a,2}}_\xi} + 
\abs{\expect{\A_{3a,3}}_\xi}
\leq 
C a^4 k_F^6 \abs{\log ak_F} \expect{\mcN}_\xi^{1/2} \expect{\mcN_>}_\xi^{1/2}\leq 
C a^4 k_F^6 \abs{\log ak_F} \expect{\mcN}_\xi .
\end{equation}

\paragraph{---Terms of the form $b^*b$:}
These may be dealt with as the terms $\A_{3a,1}$,  $\A_{3a,2}$ and $\A_{3a,3}$ only we bound $\norm{v}_{L^\infty}\leq C k_F^3$ instead of $\norm{c}^2\leq C k_F^3$.

\paragraph{---Terms of the form $c^*c$:}
For these terms it is important to take into account the cancellations between different terms. 
Recalling the formulas for the commutators in \Cref{eqn.calc.[bb.bb],eqn.calc.[cc.cc]} and noting the symmetries $x \leftrightarrow y$ and $z\leftrightarrow z'$
we find that together all terms of this form are given by 
\begin{equation*}
\begin{aligned}
\A_{3c} & = 2\iiiint V(x-y) \varphi(z-z') 
[u^r(x-z')u^r(y-z) - u^r(x-z)u^r(y-z')]
  v(y-z) 
\\ & \hphantom{=2\iiiint}  \times 
  c_x^* c_{z'} \ud x \ud y \ud z \ud z'.
\end{aligned}
\end{equation*}
Writing the term in momentum-space we find 
\begin{equation*}
\A_{3c} = \frac{2}{L^6} \sum_{\substack{k,\ell,\ell' \\ \ell,\ell'\in B_F}} \hat V(k) 
  \left[\hat \varphi(k) - \hat\varphi(k + \ell-\ell')\right] \hat u^r(k-\ell) \hat u^r(k+\ell') c_\ell^* c_\ell.
\end{equation*}
We view $\hat\varphi$ as being defined on all of $\R^3$ (as opposed to just on $\frac{2\pi}{L}\Z^3$)
and Taylor expand. 
Noting that $\varphi$ is a real radial function the first order vanishes. That is,
\begin{equation*}
\begin{aligned}
\hat\varphi(k + \ell-\ell') 
	& = \hat\varphi(k) + (\ell-\ell')^\mu (\ell-\ell')^\nu \int_0^1 \ud t \, (1-t) \nabla^\mu \nabla^\nu \hat\varphi(k + t(\ell-\ell')).
	% \\
\end{aligned}
\end{equation*}
We write $\hat u^r(k) = 1 - \hat v^r(k)$ and observe that $\hat v^r$ is supported for $|k| \leq 3k_F$. 
Since further $\ell,\ell'\in B_F$ we have $\hat v^r(k + \ell'), \hat v^r(k-\ell)\leq \chi_{|k|\leq 4 k_F}$.
The terms with at least one factor $\hat v^r$ we then bound by 
\begin{equation*}
\begin{aligned}
  & 2 \int_0^1 \ud t \, (1-t) \frac{1}{L^6} \sum_{\ell,\ell'\in B_F} |\ell-\ell'|^2 \expect{c_\ell^*c_\ell}_\xi 
    \sum_{|k|\leq 4 k_F} |\hat V(k)| |\nabla^2 \hat \varphi(k + t(\ell-\ell'))|
  \\
  & \quad \leq C k_F^{8} \norm{V}_{L^1} \norm{ |\cdot|^2 \varphi}_{L^1}  \sum_{\ell} \expect{c_\ell^* c_\ell}_\xi 
  \leq C a^4 k_F^6 \expect{\mcN}_\xi.
\end{aligned}
\end{equation*}
We compute the term with both factors $1$ using Parseval's theorem. It is given by 
\begin{equation*}
	- 2 \int_0^1 \ud t \, (1-t) \frac{1}{L^3} \sum_{\ell,\ell'\in B_F} (\ell-\ell')^\mu (\ell-\ell')^\nu \expect{c_\ell^* c_\ell}_\xi 
	\int V(x) x^\mu x^\nu e^{-it(\ell-\ell')x}\varphi(x) \ud x.
\end{equation*}
Using $\varphi \leq 1$ we bound this by $C k_F^5 \norm{|\cdot|^2 V}_{L^1} \expect{\mcN}_\xi \leq C a^3 k_F^5 \expect{\mcN}_\xi$.
We conclude that 
\begin{equation}\label{eqn.Q2B.A.3c}
\begin{aligned}
\abs{\expect{\A_{3c}}_\xi}
	& \leq C a^3 k_F^5 \expect{\mcN}_\xi.
\end{aligned}
\end{equation}

\paragraph{---Constant terms:}
The constant term is given by 
\begin{equation*}
\begin{aligned}
\A_{3d} 
  & = 
  \frac{1}{2}
  \iiiint \ud x \ud y \ud z \ud z' \, V(x-y) \varphi(z-z') 
    \Bigl[ 
      	% \left(
          \delta(x-z)\delta(y-z') - \delta(x-z')\delta(y-z) 
        % \right)
  \\ & \quad 
      - 
      % \left(
        u^r(x-z)u^r(y-z') 
        + u^r(x-z')u^r(y-z) 
      % \right)
    \Bigr]  
    % \times 
    \left(
      v(x-z)v(y-z') - v(x-z')v(y-z)
    \right).
\end{aligned}
\end{equation*}
We write $u^r = \delta - v^r$ and expand everything. 
The terms with one factor $v^r$ and one factor $\delta$ together give 
\begin{equation*}
\A_{3d,1}
	= 2 \iiint V(x-y) \varphi(x-z) v^r(y-z) 
    \left(
      v(0)v(y-z) - v(x-z)v(y-x)
    \right)
    \ud x \ud y \ud z.
\end{equation*}
Noting that the last factor 
vanishes for $x=y$ we Taylor expand and bound it as 
\begin{equation*}
\abs{ v(0)v(y-z) - v(x-z)v(y-x)} \leq C k_F^{7} |x-y|
\end{equation*}
using that $\|\nabla v\|_{L^\infty} \leq C k_F^4$. 
Then 
\begin{equation*}
\begin{aligned}
\abs{\A_{3d,1}} 
	& \leq C k_F^7 \iiint |x-y| V(x-y) \varphi(x-z) |v^r(y-z)| \ud x \ud y \ud z
	\\ & 
	\leq C L^3 k_F^{7} \norm{|\cdot|V}_{L^1}\norm{\varphi}_{L^1} \norm{v^r}_{L^\infty}
	\leq C N a^5 k_F^7 \abs{\log ak_F}.
\end{aligned}
\end{equation*}
Similarly all terms with two factors $v^r$ are together bounded by 
\begin{equation*}
C L^3 k_F^{7} \norm{|\cdot|V}_{L^1}\norm{\varphi}_{L^1} \norm{v^r}_{L^2}^2
	\leq C N a^5 k_F^7 \abs{\log ak_F}.
\end{equation*}
In particular, we conclude that 
\begin{equation}\label{eqn.Q2B.A.3d}
\abs{\A_{3d}}
  \lesssim N a^5 k_F^7 \abs{\log ak_F}.
\end{equation}

Combining then 
Equations~\eqref{eqn.Q2B.A.1.1}--\eqref{eqn.Q2B.A.3d}
we find 
\begin{equation*}
\begin{aligned}
\abs{\expect{\Q_{2;B}^\mcE}_\xi}  
  & \lesssim 
  	a^2 k_F^3 \expect{\mcN_>}_\xi^{1/2} \expect{\Q_4}_{\xi}^{1/2} 
  	+ a^3 k_F^5 \expect{\mcN}_\xi
  	+ N a^5 k_F^7 \abs{\log ak_F}.
\end{aligned}
\end{equation*}
Finally, using  \eqref{eqn.bdd.N>.a.priori} to bound $\mcN_>$ we conclude the desired.

\section{Finalizing the Proof}\label{sec.proof.final}
In this section, we use the bounds of \Cref{lemma.bdd.list.operators} to prove \Cref{eqn.mcE_scat.main,eqn.mcE_Q2.main}.
We shall also give also the proof of \Cref{prop.upper.bdd}.

\subsection{Propagation of a priori estimates}
To use the bounds of \Cref{lemma.bdd.list.operators} we need bounds for the expectation value of $\mcN$, $\H_0$ and $\Q_4$ in the states $\xi_\lambda$. 
These states do not necessarily arise from approximate ground states, 
so we cannot just use the a priori bounds of \Cref{lem.a.priori.H0,lem.a.priori.Q4,lem.a.priori.mcN}.
We show (as in \cite[Proposition 4.11]{Giacomelli.2023})

\begin{lemma}[{Propagation of $\mcN$}]\label{lem.bdd.mcN.xi_lambda}\label{lem.propagation.mcN}
Let $\psi \in \mcF$ be any state and let $\xi_\lambda$ be as in \Cref{eqn.def.xi.lambda}.
Then, for any $0 \leq \lambda,\lambda' \leq 1$,
\begin{equation*}
\expect{\mcN}_{\xi_\lambda} \leq C \expect{\mcN}_{\xi_{\lambda'}} + C N (ak_F)^{5}.
\end{equation*}
\end{lemma}

\begin{proof}
This is an application of Grönwall's lemma. We compute 
\begin{equation*}
\dd{\lambda} \expect{\mcN}_{\xi_\lambda} = - \longip{\xi_\lambda}{[\mcN,B]}{\xi_\lambda}
  = -4 \Re \iint \varphi(z-z') \longip{\xi_\lambda}{b_z^r b_{z'}^r c_{z'}c_z}{\xi_\lambda} \ud z \ud z'.
\end{equation*}
Using \Cref{eqn.bdd.phi*bcc} 
to bound the integral in $z'$ we find
\begin{equation*}
\abs{\dd{\lambda} \expect{\mcN}_{\xi_\lambda}}
\leq C k_F^{3/2} (ak_F)^{5/2} \int \norm{b_z^r \xi_\lambda} \ud z 
  \leq C N^{1/2} (ak_F)^{5/2} \expect{\mcN}_{\xi_\lambda}^{1/2}.
\end{equation*}
By Grönwall's lemma we conclude the desired.
\end{proof} 

\begin{remark}
Using \Cref{lem.bdd.mcN.xi_lambda} for an approximate ground state $\psi$ for $\lambda'=0$ we find (for any $0\leq \lambda \leq 1$)
\begin{equation}\label{eqn.bdd.mcN.xi_lambda}
\expect{\mcN}_{\xi_\lambda} \leq C \expect{\mcN}_{\xi_0} + C N (ak_F)^{5} \leq C N (ak_F)^{3/2}
\end{equation}
by \Cref{lem.a.priori.mcN}. 
\end{remark}

\begin{lemma}[{Propagation of $\H_0,\Q_4$}]\label{lem.propagation.H0.Q4}
Let $\psi\in \mcF$ be any state and let $\xi_\lambda$ be as in \Cref{eqn.def.xi.lambda}. 
Then, for any $\alpha>0$ and any $0\leq \lambda, \lambda' \leq 1$,
\begin{equation*}
\begin{aligned}
\expect{\H_0 + \Q_4}_{\xi_\lambda} 
	&\lesssim \expect{\H_0 + \Q_4}_{\xi_{\lambda'}} + 
		 N a^3 k_F^5 \left[1 + (ak_F)^{5/2 - 3\alpha/2 }\abs{\log ak_F}\right]
 \\ & \quad 
  + k_F^2(ak_F)^{6} \expect{\mcN}_{\xi_{\lambda'}} + N^{1/2}k_F^2 (ak_F)^{3-3\alpha/2} \abs{\log ak_F}  \expect{\mcN}_{\xi_{\lambda'}}^{1/2},
\end{aligned}
\end{equation*}
the implicit constants depending only on $\alpha$.
\end{lemma}

\begin{remark}
We apply \Cref{lem.propagation.H0.Q4} for an approximate ground state $\psi$ and choose $\alpha < 1/2$ and $\lambda'=0$. 
Noting that both $\H_0$ and $\Q_4$ are positive operators we find (for any $0\leq \lambda \leq 1$)
\begin{equation}\label{eqn.bdd.H0.Q4.xi.lambda}
\expect{\H_0}_{\xi_\lambda} \leq  C N a^3 k_F^5,
\qquad 
\expect{\Q_4}_{\xi_\lambda} \leq  C N a^3 k_F^5
\end{equation}
by \Cref{lem.a.priori.H0,lem.a.priori.Q4,lem.a.priori.mcN}. 
\end{remark}

\begin{proof}
Again, this is an application of Grönwall's lemma. 
We have 
\begin{equation*}
\dd{\lambda} \longip{\xi_\lambda}{\H_0+\Q_4}{\xi_\lambda} 
  = - \longip{\xi_\lambda}{[\H_0+\Q_4,B]}{\xi_\lambda}
  = - \longip{\xi_\lambda}{[\H_0+\Q_4,B] + \Q_2}{\xi_\lambda} + \longip{\xi_\lambda}{\Q_2}{\xi_\lambda}.
\end{equation*}
We may bound the first term using \Cref{eqn.[H0.B]+Q2.decompose} and the bounds of \Cref{lemma.bdd.list.operators}. This way we obtain
\begin{equation*}
\begin{aligned}
& \abs{\longip{\xi_\lambda}{[\H_0+\Q_4,B] + \Q_2}{\xi_\lambda}} 
 \\ & 
 \quad 
 \lesssim 
  N^{1/2} k_F^2 
    (ak_F)^{3 - 3\alpha/2}
    \abs{\log ak_F} 
    \expect{\mcN}_{\xi_\lambda}^{1/2}
  + N^{1/2} k_F [ (ak_F)^{3/2 + \alpha} + (a k_F)^3 \abs{\log ak_F}]  \expect{\H_0}_{\xi_\lambda}^{1/2}
  \\ & 
  \qquad 
  % \quad 
  + N^{1/2} k_F (ak_F)^{3+1/2} \expect{\Q_4}_{\xi_\lambda}^{1/2} 
  + a^3 k_F^4 \expect{\mcN}_{\xi_\lambda}^{1/2} \expect{\Q_4}_{\xi_\lambda}^{1/2}
  \\ & 
  \quad 
  \lesssim  \expect{\Q_4}_{\xi_\lambda} 
  +    \expect{\H_0}_{\xi_\lambda} +
  N k_F^2 \left[(ak_F)^{3+2 \alpha} + (ak_F)^{3+5/2-3\alpha/2} \abs{\log ak_F} \right] 
  \\ & \qquad 
  + k_F^2(ak_F)^{6} \expect{\mcN}_{\xi_{\lambda'}} + N^{1/2}k_F^2 (ak_F)^{3-3\alpha/2} \abs{\log ak_F}  \expect{\mcN}_{\xi_{\lambda'}}^{1/2}
\end{aligned}
\end{equation*}
for any $\alpha>0$ and $0\leq \lambda, \lambda'\leq 1$, where we have used \Cref{lem.propagation.mcN} in the last step.
Bounding $\abs{\expect{\Q_2}_{\xi_\lambda}} \leq \expect{\Q_4}_{\xi_\lambda} + C N a^3k_F^5$ 
by \Cref{lem.a.priori.Q2} 
we obtain
\begin{equation*}
\begin{aligned}
 \abs{\longip{\xi_\lambda}{[\H_0+\Q_4,B]}{\xi_\lambda}}  
	% \\ & \quad  
&	\lesssim \expect{\H_0 + \Q_4}_{\xi_\lambda} 
	% \\ & \quad 
+ N a^3 k_F^5 \left[1 + (ak_F)^{5/2 - 3\alpha/2 }\abs{\log ak_F}\right]
 \\ & \quad 
  + k_F^2(ak_F)^{6} \expect{\mcN}_{\xi_{\lambda'}} + N^{1/2}k_F^2 (ak_F)^{3-3\alpha/2} \abs{\log ak_F}  \expect{\mcN}_{\xi_{\lambda'}}^{1/2}
  .
\end{aligned}
\end{equation*}
By Grönwall's lemma we conclude the desired.
\end{proof}

\subsection{Lower bound}
We now give the proof of \Cref{prop.mcE_Q2,prop.mcE_scat}, therefore concluding the proof of \Cref{thm.main}.

\begin{proof}[{Proof of \Cref{prop.mcE_Q2}}]
Combining \Cref{lemma.bdd.list.operators} with the bounds on $\mcN, \H_0$ and $\Q_4$ from \Cref{eqn.bdd.H0.Q4.xi.lambda,eqn.bdd.mcN.xi_lambda} 
we find for any approximate ground state $\psi$ that 
\begin{align*}
\abs{\expect{\Q_{2;B}^\mcE}_{\xi_\lambda}}
& \lesssim 
N a^4 k_F^6 (ak_F)^{1/2}
\end{align*}
with $\xi_\lambda$ as in \Cref{eqn.def.xi.lambda}. 
Recalling  \Cref{eqn.def.mcE_Q2}, this concludes the proof of \Cref{prop.mcE_Q2}. 
\end{proof}

\begin{proof}[{Proof of \Cref{prop.mcE_scat}}]
As in the proof of \Cref{prop.mcE_Q2} above 
we find for any approximate ground state $\psi$ and any $\alpha > 0$
\begin{align*}
\abs{\expect{\T}_{\xi_\lambda}}
+ \abs{\expect{\Q_{\textnormal{scat}}}_{\xi_\lambda}} 
& \lesssim 
N a^3 k_F^{5} 
\left[(ak_F)^{3/4-3\alpha/2} \abs{\log ak_F}  + (ak_F)^{\alpha}\right],
\\
\abs{\expect{\H_{0;B}^{\div r}}_{\xi_\lambda}}
& \lesssim 
N a^3 k_F^5 (ak_F)^{3/4} \abs{\log ak_F},
\\
\abs{\expect{\Q_{4;B}^\mcE}_{\xi_\lambda}}
& \lesssim 
N a^3 k_F^5 (ak_F)^2.
\end{align*}
In particular, recalling \Cref{eqn.def.mcE_scat,eqn.[H0.B]+Q2.decompose}, we have the bound
\begin{equation*}
\begin{aligned}
\abs{\mcE_{\textnormal{scat}}(\psi)}
  & \leq C N a^3 k_F^{5} 
\left[(ak_F)^{3/4-3\alpha/2} \abs{\log ak_F}  + (ak_F)^{\alpha}\right]
\end{aligned}
\end{equation*}
 for any $\alpha >0$. 
Choosing the optimal $\alpha = 3/10$ we conclude the proof of \Cref{eqn.mcE_scat.main}.
\end{proof}

\subsection{Upper bound}
Finally, we prove  the upper bound stated in \Cref{prop.upper.bdd}. The starting point is again \Cref{eqn.main}. To obtain an upper bound, 
 we need to choose some trial state $\psi$ and compute the energy. 
For the lower bound we used that $\H_0 + \Q_4 \geq 0$, 
whereas in general we can only bound $\expect{\H_0 + \Q_4}_{\xi_1} \leq C N a^3 k_F^5$ for an approximate ground state $\psi$ 
(recall the definition of $\xi_\lambda$ in \Cref{eqn.def.xi.lambda}).
This bound is not good enough for our purposes. 
We can, however, choose the specific trial state $\psi = RT\Omega$, satisfying  
$\xi_1 = T^{-1}R^*\psi = \Omega$ and thus $\expect{\H_0 + \Q_4}_{\xi_1} = 0$.
However, it is not  a priori clear whether or not this trial state is an approximate ground state.
Thus, we cannot just apply the bounds on the error terms $\mcE_V(\psi), \mcE_{\Q_2}(\psi)$ and $\mcE_{\textnormal{scat}}(\psi)$
from \Cref{prop.mcE_scat,prop.mcE_Q2,prop.mcE_V}.
Instead, we bound these error terms as follows.

We first note the bounds (for any $0\leq \lambda \leq 1$)
\begin{equation}\label{lem.bdd.xi.lambda.upper.bdd}
  \expect{\mcN}_{\xi_\lambda} \leq C N a^5 k_F^5,
  \qquad 
  \expect{\H_0}_{\xi_\lambda} \leq C N a^3 k_F^5,
  \qquad 
  \expect{\Q_4}_{\xi_\lambda} \leq C N a^3 k_F^5
\end{equation}
by applying Lemmas~\ref{lem.propagation.mcN} and \ref{lem.propagation.H0.Q4}
for $\lambda' = 1$ and using that $\xi_1 = \Omega$.
Using these bounds, \Cref{lem.ini.V=Q2+Q4} reads 
\begin{equation*}
  \abs{\mcE_V(\psi)} \lesssim \eps N a^3 k_F^5 
    + \eps^{-1} N a^{6} k_F^8
    \leq C N a^4 k_F^6 (ak_F)^{1/2}
\end{equation*}
by choosing the optimal $\eps = (ak_F)^{3/2}$. 
Next, using \Cref{lemma.bdd.list.operators} and the bounds in \Cref{lem.bdd.xi.lambda.upper.bdd} we find the bounds
 (recall \Cref{eqn.def.mcE_scat,eqn.def.mcE_Q2})
\begin{equation*}
\begin{aligned}
\abs{\mcE_{\Q_2}(\psi)}
  & \leq C N a^5 k_F^7 \abs{\log ak_F},
  \\
\abs{\mcE_{\textnormal{scat}}(\psi)}
  & \leq 
    C N a^5 k_F^7 (ak_F)^{1/2-3\alpha/2} \abs{\log ak_F} 
    + C N a^3 k_F^5 (ak_F)^{\alpha}
    + C N a^4 k_F^6 (ak_F)^{1/2} \abs{\log ak_F}
\end{aligned}
\end{equation*}
for any $\alpha > 0$. 
Choosing the optimal $\alpha = 1$ we find $\abs{\mcE_{\textnormal{scat}}(\psi)} \leq C N a^4 k_F^6 \abs{\log ak_F}$.
Together with the computation of $\expect{\!\ud\Gamma(V(1-\varphi))}_F$ in \Cref{eqn.calc.<Vphi>.F}, 
this concludes the proof of \Cref{prop.upper.bdd}.

\subsubsection*{Acknowledgments}
We thank 
Morris Brooks,
Emanuela Giacomelli,
Christian Hainzl,
Marcello Porta
and 
Phan Thành Nam
for helpful discussions.
Financial support by the Austrian Science Fund (FWF) through 
grant \href{https://doi.org/10.55776/I6427}{DOI: \nolinkurl{10.55776/I6427}}
% project number~I~6427-N 
(as part of the SFB/TRR~352) is gratefully acknowledged.

\appendix 

\section{The Two-Dimensional Case}\label{sec.d=2}
We sketch here how to adapt the proof of \Cref{thm.main} to the two-dimensional setting. 
The structure is the same as in the case $d=3$. 
The main step in the proof of \Cref{thm.main.d=2} is proving the two-dimensional analogue of \Cref{lemma.bdd.list.operators}. 
This reads  
\begin{lemma}\label{lemma.bdd.list.operators.d=2}
For any state  $\xi\in \mcF$, and any $\alpha > 0$, we have 
\begin{align*}
\abs{\expect{\T}_\xi}
+
\abs{\expect{\Q_{\textnormal{scat}}}_\xi}
& \lesssim 
N^{1/2} a^2 k_F^4 (ak_F)^{-\alpha} \abs{\log ak_F}  \expect{\mcN}_\xi^{1/2} 
	+ N^{1/2} a k_F^2 (ak_F)^{\alpha} \abs{\log ak_F}^{1/2} \expect{\H_0}_\xi^{1/2},
\\
\abs{\expect{\H_{0;B}^{\div r}}_\xi}
& \lesssim 
N^{1/2} a^2 k_F^4 \abs{\log ak_F} \expect{\mcN}_\xi^{1/2}
+
N^{1/2} a^2 k_F^3 \abs{\log ak_F} \expect{\H_0}_{\xi}^{1/2},
\\
\abs{\expect{\Q_{4;B}^\mcE}_\xi}
& \lesssim 
N^{1/2} a^2 k_F^3 \abs{\log ak_F} \expect{\Q_4}_\xi^{1/2}
+ a^{2}k_F^3 \abs{\log ak_F}^{1/2} \expect{\mcN}_{\xi}^{1/2} \expect{\Q_4}_{\xi}^{1/2},
\\
\abs{\expect{\Q_{2;B}^\mcE}_\xi}
& \lesssim 
a k_F  \expect{\H_0}_\xi^{1/2} \expect{\Q_4}_\xi^{1/2}
+
a^2 k_F^4 \abs{\log ak_F} \expect{\mcN}_\xi 
% \\ & \quad 
+ N a^3 k_F^5 \abs{\log ak_F}.
\end{align*}
\end{lemma}

Further, as in $3$ dimensions, we have the a priori bounds 
\begin{equation*}
\expect{\mcN}_{\xi_0} \leq C N a k_F,
\qquad 
\expect{\H_0}_{\xi_0} \leq C N a^2 k_F^4,
\qquad 
\expect{\Q_4}_{\xi_0} \leq C N a^2 k_F^4
\end{equation*}
for an approximate ground state $\psi$ (recall the definition of $\xi_\lambda$ in \Cref{eqn.def.xi.lambda}).
Using \Cref{lemma.bdd.list.operators.d=2} and these a priori bounds we give the

\begin{proof}[{Proof of \Cref{thm.main.d=2}}]
Propagating the a priori bounds as in \Cref{sec.proof.final}, using \Cref{lemma.bdd.list.operators.d=2}
and choosing the optimal $\alpha = 1/4$ we conclude the proof of \Cref{thm.main.d=2}.
\end{proof}

To prove \Cref{lemma.bdd.list.operators.d=2} we note the bounds on norms of $\varphi$: (being the analogue of \Cref{lem.prop.phi})

\begin{lemma}
The scattering function $\varphi$ satisfies
\begin{equation*}
\begin{aligned}
\norm{|\cdot|^n \varphi}_{L^1}
& \leq 
C a^2 k_F^{-n}, 
& n &= 1,2,
\quad  
&
\norm{|\cdot|^n \nabla^n \varphi}_{L^1}
& \leq C a^2 \abs{\log ak_F}, 
&n&=0,1,2
\quad  
\\
\norm{ |\cdot| \varphi}_{L^2 }
& \leq 
C a^{2} \abs{\log ak_F}^{1/2}, 
&&
\quad 
&
\norm{|\cdot|^n \nabla^n \varphi}_{L^2}
& \leq C a, 
&n&=0,1.
\end{aligned}
\end{equation*}
\end{lemma}

Further, we have the analogue of \Cref{lem.bdd.F*ac.b(F)}: (which we state in general dimension $d$)
\begin{lemma}\label{lem.bdd.F*ac.b(F).d=2}
Let $F$ be a compactly supported function with $F(x)=0$ for $|x| \geq C k_F^{-1}$.
Then, uniformly in $x\in \Lambda$ and $t\in [0,1]$, (with $\nabla^n$ denoting any $n$'th derivative)
\begin{align*}
\norm{\int F(x-y) a_y c_y \ud y} 
	& \lesssim k_F^{d/2} \norm{F}_{L^2} ,
\qquad 
\norm{\int F(x-y) a_y c_y \nabla^n c_{ty+(1-t)x} \ud y 	}
	% & 
	\lesssim k_F^{d+n} \norm{F}_{L^2}.
\end{align*}
\end{lemma}

We then sketch the

\begin{proof}[{Proof of  \Cref{lemma.bdd.list.operators.d=2}}]
The proof is the same as that of \Cref{lemma.bdd.list.operators} only with trivial changes. 
We here just state all the intermediate bounds. 
We state these in general dimension $d$, to make clear the changes needed to extend the proof also to the one-dimensional setting.

We may bound $\T, \Q_{\textnormal{scat}}$ and $\H_{0;B}^{\div r}$ by 
\begin{align*}
\abs{\expect{\T}_\xi}
	& \leq C N^{1/2} k_F^{d+2} (ak_F)^{-d\alpha/2}\norm{F}_{L^1} \expect{\mcN}_{\xi}^{1/2}
		+ C N^{1/2} k_F^{d/2+2} \norm{F}_{L^2} \expect{\mcN_{>\alpha}}_\xi^{1/2}.
\\
\abs{\expect{\Q_{\textnormal{scat}}}_\xi}
	& \leq C N^{1/2} k_F^{d+1} (ak_F)^{-d\alpha/2}\norm{\mcE_\varphi}_{L^1} \expect{\mcN}_{\xi}^{1/2}
		+ C N^{1/2} k_F^{d/2+1} \norm{\mcE_\varphi}_{L^2} \expect{\mcN_{>\alpha}}_\xi^{1/2}.
\\
\abs{\expect{\H_{0;B}^{\div r}}_\xi} 
	& \leq C N^{1/2} \left[\norm{|\cdot|^2\Delta\varphi}_{L^1} + \norm{|\cdot|\nabla \varphi}_{L^1}\right] 
	\left[k_F^{d+2} \expect{\mcN}_\xi^{1/2} + k_F^{d+1} \expect{\H_0}_\xi^{1/2}\right].
\end{align*}
For the error terms from $[\Q_4,B]$ we have the bounds 
\begin{equation*}
\begin{aligned}
\abs{\expect{\A_1}_\xi} + \abs{\expect{\A_2}_\xi}
	& 
	\leq C N^{1/2} k_F^{d+d/2+1} \norm{|\cdot|\varphi}_{L^1} \norm{V}_{L^1}^{1/2} \expect{\Q_4}_\xi^{1/2}
	\\
	% \qquad 
\abs{\expect{\A_3}_\xi} + \abs{\expect{\A_4}_\xi}
	& 
	\leq C 
	k_F^{d+1} \norm{|\cdot|\varphi}_{L^2}
	\norm{V}_{L^1}  \expect{\mcN}_\xi^{1/2} \expect{\Q_4}_\xi^{1/2}.
\end{aligned}
\end{equation*}
Finally, the error terms from $[\Q_2,B]$ are bounded as follows. 
\begin{align*}
\abs{\expect{\A_{1a}}_\xi}
+ \abs{\expect{\A_{1b}}_\xi}
	& \leq C k_F^d
	\norm{\varphi}_{L^2}
	\norm{V}_{L^1}^{1/2} \expect{\mcN_>}_{\xi}^{1/2} \expect{\Q_4}_{\xi}^{1/2}
	\\
\abs{\expect{\A_{2a}}_\xi}
	& \leq C k_F^{2d+2} \norm{|\cdot|\varphi}_{L^1} \norm{|\cdot|V}_{L^1} \expect{\mcN}_\xi
	\\
\abs{\expect{\A_{2b}}_\xi}
	& \leq C k_F^{d+2} \norm{|\cdot|^2V}_{L^1} \expect{\mcN}_\xi  
	+ C k_F^{d+d/2+2} \norm{|\cdot|^2V}_{L^1}^{1/2} \norm{V}_{L^1}^{1/2} \norm{|\cdot|\varphi}_{L^2} \expect{\mcN}_\xi 
	\\ & \quad 
	+ C k_F^{2d+2} \norm{|\cdot|^2V}_{L^1}^{1/2} \norm{V}_{L^1}^{1/2} \norm{|\cdot|^2\varphi}_{L^1}^{1/2} \norm{\varphi}_{L^1}^{1/2} \expect{\mcN}_\xi 
	\\
\abs{\expect{\A_{3a}}_\xi} + \abs{\expect{\A_{3b}}_\xi}
	& \leq C k_F^{2d} \norm{\varphi}_{L^1}\norm{V}_{L^1} \expect{\mcN}_\xi
	\\
\abs{\expect{\A_{3c}}_\xi} 
	& \leq 
	C k_F^{2d+2} \norm{V}_{L^1} \norm{|\cdot|^2\varphi}_{L^1} \expect{\mcN}_\xi 
	+ C k_F^{d+2} \norm{|\cdot|^2 V}_{L^1} \expect{\mcN}_\xi
	\\
\abs{\expect{\A_{3d}}_\xi}
	& \leq C N k_F^{2d+1} \norm{|\cdot|V}_{L^1} \norm{\varphi}_{L^1}.
\end{align*}
Bounding $\mcN_>$ and $\mcN_{>\alpha}$ as in \Cref{eqn.bdd.N>.a.priori} and using the bounds on norms of $\varphi$ above we conclude the proof of \Cref{lemma.bdd.list.operators.d=2}.
\end{proof}

\begin{remark}[{Possible extension to the one-dimensional case}]\label{rmk.1d.d=1}
In dimension $d=1$ the analogous bounds on norms of $\varphi$ read
\begin{equation*}
\begin{aligned}
\norm{|\cdot|^n \varphi}_{L^1}
& \leq 
C a k_F^{-n}, 
& n &= 1,2,
\quad  
&
\norm{|\cdot|^n \nabla^n \varphi}_{L^1}
& \leq C a \abs{\log ak_F}, 
&n&=0,1,2
\quad  
\\
\norm{ |\cdot| \varphi}_{L^2 }
& \leq 
C a k_F^{-1/2}, 
&&
\quad 
&
\norm{|\cdot|^n \nabla^n \varphi}_{L^2}
& \leq C a^{1/2}, 
&n&=0,1.
\end{aligned}
\end{equation*}
Further $\norm{|\cdot|^n V}_{L^1}\leq C a^{n-1}$.
Using these bounds and bounding $\expect{\mcN}_\xi \leq C N (ak_F)^{1/2}$, $\expect{\H_0}_\xi \leq C N a k_F^3$ and $\expect{\Q_4}_\xi \leq C N a k_F^3$
(as is appropriate for relevant $\xi$ by the propagation of the a priori bounds)
we see that only the estimates of the error terms $\T, \Q_{\textnormal{scat}}, \H_{0;B}^{\div r},\A_{2a}$ and $\A_{2b}$ 
in the proof or \Cref{lemma.bdd.list.operators.d=2} are smaller than $N a k_F^3$. 
The estimates of the (many) remaining error-terms are too big.
One would have to improve these remaining estimates in order to extend our proof of \Cref{thm.main,thm.main.d=2} also to the one-dimensional setting. 
We do not pursue this. As already mentioned, the one-dimensional analogue of \Cref{thm.main,thm.main.d=2} is proved in \cite{Agerskov.Reuvers.ea.2022}. 
\end{remark}

\printbibliography

\end{document}